\documentclass{article}
\usepackage[utf8]{inputenc}
\usepackage{fullpage}

\usepackage{algpseudocode, algorithm}
\usepackage{amssymb}
\usepackage{amsmath}
\usepackage{amsthm}
\usepackage{amsfonts}
\usepackage{mathtools}
\usepackage{upgreek}
\usepackage{xspace}
\usepackage[charter]{mathdesign}
\usepackage{eulervm}
\usepackage{enumitem}
\usepackage{boxedminipage}
\usepackage[x11names]{xcolor}
\usepackage{graphicx}
\usepackage{framed}
\usepackage{caption}
\usepackage{ifthen}
\usepackage{url}
\usepackage{verbatim}
\usepackage{array}
\usepackage{hyperref}
\def\final{0}  % set this to 1 to get a comment-free version
\def\iflong{\iffalse}
\ifnum\final=0  %namely if we allow comments in the output
\newcommand{\jnote}[1]{{\color{red}[\small {Jeremiah: \bf #1}]\marginpar{\color{red}*}}}
\newcommand{\enote}[1]{{\color{green}[{\small Elena: \bf #1}]\marginpar{\color{red}*}}}
\newcommand{\tanote}[1]{{\color{purple}[\small {Tamalika: \bf #1}]\marginpar{\color{red}*}}}
\else % in this case [final=1] we don't want any comments to show
\newcommand{\jnote}[1]{}
\newcommand{\enote}[1]{}
\newcommand{\tanote}[1]{}
\newcommand{\todo}[1]{}
\fi

  \newcommand{\Lap}{\ensuremath{{\text{Lap}}}\xspace}

\newcommand{\E}[0]{\mathop{\bbE}\xspace}

\newtheorem{fact}{Fact}

  \newcommand{\eps}[0]{\ensuremath{\varepsilon}}
  \let\epsilon\eps
  
  \newcommand{\cA}{\ensuremath{{\mathcal A}}\xspace}

  \newcommand{\cD}{\ensuremath{{\mathcal D}}\xspace}

  \newcommand{\cM}{\ensuremath{{\mathcal M}}\xspace}
  
  \newcommand{\cO}{\ensuremath{{\mathcal O}}\xspace}

  \newcommand{\cR}{\ensuremath{{\mathcal R}}\xspace}
  \newcommand{\cS}{\ensuremath{{\mathcal S}}\xspace}

  \newcommand{\cX}{\ensuremath{{\mathcal X}}\xspace}
  
  \newcommand{\cZ}{\ensuremath{{\mathcal Z}}\xspace}

  \newcommand{\bbE}{\ensuremath{{\mathbb E}}\xspace}

  \newcommand{\bbR}{\ensuremath{{\mathbb R}}\xspace}

\makeatletter
\newtheorem*{rep@theorem}{\rep@title}
\newcommand{\newreptheorem}[2]{%
\newenvironment{rep#1}[1]{%
 \def\rep@title{\theoremref{##1} Restated}%
 \begin{rep@theorem}}%
 {\end{rep@theorem}}}
\makeatother

\makeatletter
\newtheorem*{rep@lemma}{\rep@title}
\newcommand{\newreplemma}[2] {%
\newenvironment{rep#1}[1]{%
 \def\rep@title{\lemmaref{##1} Restated}%
 \begin{rep@lemma}}%
 {\end{rep@lemma}}}
\makeatother

\newtheorem{theorem}{Theorem}
\newreptheorem{theorem}{Theorem}

\newtheorem{definition}{Definition}
\newtheorem{lemma}{Lemma}

\newreplemma{lemma}{Lemma}
\newtheorem{claim}{Claim}
\newtheorem{corollary}[theorem]{Corollary}

\theoremstyle{definition}
\newtheorem*{remark}{Remark}
\newcommand{\namedref}[2]{\hyperref[#2]{#1~\ref*{#2}}\xspace}

\newcommand{\lemmaref}[1]{\namedref{Lemma}{lem:#1}}

\newcommand{\theoremref}[1]{\namedref{Theorem}{thm:#1}}

\title{Privately Estimating Graph Parameters in Sublinear time}

\author{Jeremiah Blocki, Elena Grigorescu, Tamalika Mukherjee\\ 
Department of Computer Science, Purdue University.\\ \{jblocki, elena-g, tmukherj\}@purdue.edu
\thanks{J. B. and T.M were supported in part by NSF CNS-1931443 and NSF CCF-1910659. E.G and T. M.  were supported in part by NSF CCF-1910659 and NSF CCF-1910411.}
 }

\begin{document}

\maketitle

\begin{abstract}

We initiate a systematic study of algorithms that are both differentially-private and run in sublinear time for several problems in which the goal is to estimate natural graph parameters.
Our main result is a differentially-private $(1+\rho)$-approximation algorithm for the problem of computing the average degree of a graph, for every $\rho>0$. The running time of the algorithm is roughly the same as its non-private version proposed by Goldreich and Ron (Sublinear Algorithms, 2005).
We also obtain the first differentially-private sublinear-time approximation algorithms for the  maximum matching size and the minimum vertex cover size of a graph.

An overarching technique we employ is the notion of \emph{coupled global sensitivity} of randomized algorithms. Related variants of this notion of sensitivity have been used in the literature in ad-hoc ways. Here we formalize the notion and develop it as a unifying framework for privacy analysis of randomized approximation algorithms. 

\end{abstract}

\section{Introduction}
 Graphs are frequently used to model massive data sets (e.g., social networks) where the users are the nodes, and their relationships are the edges of the graphs. These relationships often consist of sensitive information, which drives the need for privacy in this setting. 

Differential Privacy (DP)~\cite{Dwork_McSherry_Nissim_Smith_2017} has become the gold standard in privacy-preserving data analysis due to its compelling privacy guarantees and mathematically rigorous definition. Informally, a randomized function computed on a graph is {\em differentially private} if the distribution of the function's output does not change significantly with the presence or absence of an individual edge (or node). See~\cite{DR14} for a comprehensive tutorial on differential privacy.

\begin{definition}[Differential-privacy]\label{def:DP}
Let $\mathcal{G}_n$ denote the set of all $n$-node graphs. An algorithm $\cA$ is $(\eps,\delta)$ node-DP (resp. edge-DP) if for every pair of node-neighboring (resp. edge-neighboring)\footnote{Graphs $G_1=(V, E_1)$, $G_2=(V, E_2)$ are \emph{node-neighboring}, denoted by $G_1\sim_v G_2$, if there exists a vertex $v\in V$ such that $E_1(V\setminus\{v\})=E_2(V\setminus\{v\})$. Graphs $G_1$ and $G_2$ are \emph{edge-neighboring} i.e., $G_1 \sim_e G_2$ if there exists an edge $e$ such that $E_1 \setminus{\{e\}}=E_2 \setminus{\{e\}}$.} graphs $G_1,G_2 \in \mathcal{G}_n$, and for all sets $\cS$ of possible outputs, we have that $\Pr[\cA(G_1) \in \cS ] \leq e^\eps \Pr [\cA(G_2) \in \cS] + \delta $. When $\delta=0$ we simply say that the algorithm is $\epsilon$-DP.
\end{definition}

Since the graphs appearing in modern applications are massive, it is also often desirable to design {\em sublinear-time} algorithms that approximate natural combinatorial properties of the graph, such as the average degree, the number of connected components, the cost of a minimum spanning tree, the number of triangles, the size of a maximum matching, the size of a minimum vertex cover, etc.
For an excellent survey on sublinear-time algorithms for approximating graph parameters, we refer the reader to \cite{Ron19}.

There has been a lot of work in developing differentially-private algorithms for estimating graph parameters in polynomial-time, with respect to {\em edge differential privacy}, i.e., neighboring graphs that differ by a single edge in Definition~\ref{def:DP}. Nissim, Raskhodnikova, and Smith~\cite{NRS07} demonstrated the first edge-differentially private graph algorithms. They showed how to estimate the cost of a minimum spanning tree and the number of triangles in a graph by calibrating noise to a local variant of sensitivity called {\em smooth sensitivity}. Subsequent works in designing edge differentially-private algorithms for computing graph statistics include~\cite{karwa2011private,hay2009accurate, LG2014, ZCP15}. Gupta, Ligett, McSherry, Roth and Talwar~\cite{gupta2010differentially} gave the first edge differentially-private algorithms for classical graph optimization problems, such as vertex cover, and minimum s-t cut, by making clever use of the exponential mechanism in existing non-private algorithms that solve the same problem. 

An even more desirable notion of privacy in graphs is the notion of {\em node differential privacy} i.e., neighboring graphs that differ by a single node and edges incident to it in Definition~\ref{def:DP}. The concept of node differentially-private algorithms for $1$-dimensional functions (functions that output a single real value) on graphs was first rigorously studied independently by Kasiviswanathan, Nissim, Raskhodnikova and Smith~\cite{kasiviswanathan2013analyzing}, as well as, Blocki, Blum, Datta, and Sheffet~\cite{blocki2013differentially}, and Chen and Zhou~\cite{Chen_2013}. Their techniques were later extended to higher-dimensional functions on graphs~\cite{raskhodnikova2015efficient,BCS2015}. Subsequent works have focused on developing node differentially-private algorithms for a family of network models: stochastic block models and graphons~\cite{Borgs_2018,sealfon2019efficiently}. A more recent line of work has focused on the continual release of graph statistics such as degree-distributions and subgraph counts in an online setting~\cite{SLMVC18,fichtenberger2021differentially}. Gehrke, Lui, and Pass~\cite{gehrke2011towards} introduce a more robust notion of differential privacy called Zero-Knowledge Differential Privacy (ZKDP), which tackles the problem of auxiliary information in social networks. This work uses existing results from sublinear-time algorithms as a building block to achieve ZKDP for several graph problems. However, it is important to note that the final ZKDP mechanisms are not computable in sublinear-time. 

The literature on designing differentially-private algorithms for estimating graph parameters in sublinear time is far less developed. The only paper we are aware of is due to Sivasubramaniam, Li and He~\cite{sivasubramaniamdifferentially}, who give the first sublinear-time differentially-private algorithm for approximating the average degree of a graph. %., achieving a $(2+\rho+o(1))$-approximation for every constant $\rho>0$.
%\footnote{Recently Sivasubramaniam, Li and He~\cite{sivasubramaniamdifferentially} gave a sublinear-time differentially-private algorithm for approximating the average degree of a graph that achieves a $(2+\rho+O(\rho/(\eps \sqrt{n})))$-approximation for every constant $\rho>0$.}. 
Our work addresses this gap by initiating a systematic study of differentially-private sublinear-time algorithms for the problems of estimating the following graph parameters: (1) the average-degree of a graph, (2) the size of a maximum matching, and (3) the size of a minimum vertex cover. As an overarching technique, we formally introduce the notion of {\em Coupled Global Sensitivity} and use it to analyze the privacy of our randomized approximation algorithms.

%as a technique for adding an appropriate amount of noise to our algorithms. 

%\tanote{$\cA(D)$ is a distribution on the output, so instead of coupling random coins, we are coupling the outputs of $\cA(D)$ and $\cA(D')$?}
\subsection{Our Results}

\subsubsection{Privately Approximating the Average Degree}

 We obtain a differentially-private sublinear-time algorithm for estimating the average degree  ${\bar d}_G=\frac{\sum_{v\in V} \deg(v)}{|V|}$, of a graph $G=(V, E)$, with respect to edge-differential privacy, which achieves a multiplicative approximation of $(1+\rho)$, for any constant $\rho >0$. Specifically, our algorithm outputs a value ${\tilde d}$ such that w.h.p. we have $(1-\rho){\bar d}_G\leq {\tilde d}\leq (1+\rho){\bar d}_G,$ for graphs with ${\bar d}_G=\Omega(1).$ Throughout the paper we denote $|V|=n.$
 
 We work in the {\em neighbor-query} model, in which we are given oracle access to a  simple graph $G=(V,E)$, where the algorithm can obtain the identity of the $i$-th neighbor of a vertex $v \in V$ in constant time. If $i> \deg(v)$ for a particular vertex $v$, then $\perp$ is returned. The algorithm may also perform {\em degree queries}, namely for any $v\in V$ it can obtain $\deg(v)$ in constant time.

\begin{theorem}\label{thm:informal-avg-deg}
There is an $\eps$-edge differentially-private $(1+\rho)$-approximation algorithm for estimating the average degree ${\bar d}_G\geq 1$ \footnote{Observe that for ${\bar d}_G=o(1)$ a multiplicative approximation algorithm that can distinguish between two graphs on $n$ vertices, one with $0$ edges, and another with, say $1$ edge, must  sample $\Omega(n)$ vertices, and hence cannot be running in sublinear time.} of a graph $G$ on $n$ vertices that runs in time\footnote{from here on, we use running time and number of queries interchangeably.}  $O(\sqrt{n} \cdot \text{poly}(\log (n)/\rho) \cdot \text{poly}(1/\eps))$ where $\eps^{-1} =  o(\log^{1/4}(n))$.
\end{theorem}

The problem of estimating the average degree of a graph was first studied by Feige ~\cite{feige2006sums}, who gave a sublinear time $(2+\rho)$-approximation (multiplicative) for any constant $\rho >0$, making $\tilde{O}(\sqrt{n})$ degree queries, for any constant $\rho > 0$. Feige also proved that any approximation algorithm that only utilizes degree queries and obtains a $2-o(1)$-approximation requires at least $\Omega(\sqrt{n})$ queries. 
Goldreich and Ron~\cite{goldreich2004estimating} subsequently gave a $(1+\rho)$-approximation using both degree and neighbor queries, running in time $\tilde{O}(\sqrt{n}\cdot poly(1/\rho))$. This bound  is also tight, since every constant-factor approximation algorithm must make $\Omega(\sqrt{n})$ degree and neighbor queries~ \cite{goldreich2004estimating}.  A simpler analysis achieving the same bounds was given by Seshadri ~\cite{seshadhri2015simpler}. Further, Dasgupta, Kumar and Sarlós~\cite{DKS14} studied this problem in the model where access to the graph is via samples, in the context of massive networks where the number of nodes may not be known. They obtain a $(1+\rho)$-approximation that uses roughly $O(\log d_U \cdot \log \log d_U)$ samples where $d_U$ is an upper bound on the maximum degree of the graph.

In recent work, Sivasubramaniam, Li and He~\cite{sivasubramaniamdifferentially} gave a sublinear-time differentially-private algorithm for approximating the average degree of a graph using Feige's~\cite{feige2006sums} algorithm. Their algorithm  achieves a $(2+\rho+o(1))$-approximation for every constant $\rho>0$. They achieve this by calculating a tight bound for the global sensitivity of the final estimate of Feige's algorithm and adding Laplace noise with respect to this quantity appropriately. By contrast, we achieve a $(1+\rho)$-approximation for any constant $\rho > 0$ --- assuming that the privacy parameter is $\epsilon^{-1} = o(\log^{1/4} n)$.

\subsubsection{Privately Approximating the Size of a Maximum Matching and Minimum Vertex Cover}

Given an undirected graph, a set of vertex-disjoint edges is called a \emph{matching}. A matching $M$ is \emph{maximal} if $M$ is not properly contained in another matching. A matching $M$ is \emph{maximum} if for any other matching $M'$, $\vert M \vert \geq \vert M ' \vert$. A \emph{vertex cover} of a graph is a set of vertices that includes at least one endpoint of every edge of the graph. A \emph{minimum} vertex cover is a vertex cover of the smallest possible size. 
 %\paragraph{Notation.}
 For a minimization problem, we say that a value $\hat{y}$ is an $(\alpha,\beta)$-approximation to $y$ if $y \leq \hat{y} \leq \alpha y + \beta $. For a maximization problem,  we say that a value $\hat{y}$ is an $(\alpha,\beta)$-approximation to $y$ if $\frac{y}{\alpha} - \beta \leq \hat{y} \leq  y$. An algorithm $\cA$ is an $(\alpha, \beta)$-approximation for a value $V(x)$ if it computes an $(\alpha,\beta)$-approximation to $V(x)$ with probability at least 2/3 for any proper input $x$. 

For a graph $G=(V,E)$, we work in the \emph{bounded degree} model, where one can query an $i$-th neighbor ($i\in [d]$) of a vertex in constant time; denote this query as $\text{Nbr}(v,i)$. Here $d$ is the maximum degree of the graph. If $i> \deg(v)$ for a particular vertex $v$, then $\text{Nbr}(v,i)=\perp$. We also assume query access to the degree of a vertex, i.e., one can query $\deg(v)$ for any $v\in V$ in constant time.

% We give the first differentially-private sublinear-time algorithms to estimate the size of a maximum matching and vertex cover. We achieve this by first analyzing the Coupled Global Sensitivity of the non-private sublinear-time algorithms achieving a $(2,\rho n)$-approximation given by Nguyen and Onak, whose running time was later improved in a series of works by Yoshida, Yamamoto and Ito~\cite{nguyen2008constant, yoshida2012improved}, Onak, Ron, Rosen and Rubinfeld~\cite{ORR12} , and finally by Behnezhad~\cite{behnezhad21}; and then adding Laplace noise with respect to the Coupled Global Sensitivity. Our final DP algorithm simply runs the non-private approximation algorithm and then adds Laplace noise proportional to the Coupled Global Sensitivity. Thus our time complexity is identical to the non-private approximation algorithm and we show that the added Laplace noise is small enough that it preserves the approximation guarantees of the non-private approximation algorithm. 

\begin{theorem} \label{thm:main-sub-mm}
There is an $\eps$-(node and edge) differentially-private algorithm for the maximum matching problem that reports a $(2,\rho n)$-approximation with probability $1-(2/n^4 + 1/n^{192 \eps/\rho})$, and runs in expected time $\tilde{O}\left((\bar{d}+1)/\rho^2 \right)$,  where $\bar{d}$ is the average degree of the input graph.
\end{theorem}

\begin{theorem}\label{thm:main-sub-vc}
There is an $\eps$-(node and edge) differentially-private algorithm for the minimum vertex cover problem that  reports a $(2,\rho n)$-approximation with probability $1-(2/n^4 + 1/n^{96 \eps/\rho } )$, and runs in expected time $\tilde{O}\left((\bar{d}+1)/\rho^2 \right)$,  where $\bar{d}$ is the average degree of the input graph.
\end{theorem}

Typically, the privacy parameter $\eps$ is a constant, and so is the approximation parameter $\rho$, in which case the success probability in the theorems above is $1-1/(poly(n))$.

The question of approximating the size of a vertex cover in sublinear-time was first posed by Parnas and Ron~\cite{parnas2007approximating}, who obtained a $(2,\rho n)$-approximation in time $d^{O(\log d/\rho^3)}$, where $d$ is the maximum degree of the graph. Nguyen and Onak \cite{nguyen2008constant}  improved upon this result by giving a $(2,\rho n)$-approximation for the maximum matching problem, and consequently a $(2,\rho n)$-approximation for the vertex cover problem, in time $O(2^{O(d)/\rho^2})$. The result of~\cite{nguyen2008constant} was later improved by Yoshida, Yamamoto and Ito~\cite{yoshida2012improved},  who gave an ingenious analysis of the original algorithm to achieve a running time of $O(d^4/\rho^2)$. Onak, Ron, Rosen and Rubinfeld~\cite{ORR12} proposed a near-optimal time complexity of $\tilde{O}(\bar{d} \cdot \text{poly}(1/\rho))$, where $\bar{d}$ is the average degree of a graph, but Chen, Kannan, and Khanna~\cite{ChenKK20} identified a subtlety in their analysis, which proved to be crucial to their improved time complexity claim. Very recently, building on ideas from the analysis of \cite{yoshida2012improved}, Behnezhad~\cite{behnezhad21} gave a new analysis for achieving a $(2,\rho n)$-approximation to the size of maximum matching and minimum vertex cover in time $\tilde{O}((\bar{d}+1)/\rho^2)$. Behnezhad's result nearly matches the lower bound given by Parnas and Ron~\cite{parnas2007approximating}, who showed that $\Omega(\bar{d}+1)$ queries are necessary for obtaining a $(O(1),\rho n)$-estimate in the case of the maximum matching or minimum vertex cover problem. 

Our final DP algorithm simply runs the non-private approximation algorithm~\cite{behnezhad21} and then adds Laplace noise proportional to the Coupled Global Sensitivity (of the non-private algorithm). Thus, our time complexity is identical to the non-private approximation algorithm. We show that the added Laplace noise is small enough that it preserves the approximation guarantees of the non-private approximation algorithm.

\subsection{Organization} 
We define and motivate the notion of Coupled Global Sensitivity as a privacy tool in Section~\ref{sec:cgs-techniques}. Then we give a high-level overview of the techniques used for our results in Section~\ref{sec:tech-ov}. The formal privacy and accuracy analysis of Theorem~\ref{thm:informal-avg-deg} are in Sections~\ref{sec:privacy-avgdeg} and \ref{sec:accuracy-avgdeg}. The formal analysis for Theorems~\ref{thm:main-sub-mm} and \ref{thm:main-sub-vc} are in Section~\ref{sec:main-sub-mm}. We conclude with some open problems in Section~\ref{sec:open}.  

\subsection{Coupled Global Sensitivity as a Tool in Privacy analysis}\label{sec:cgs-techniques}

{\bf Background and Motivation.} Given a query $f: \mathcal{D} \to \mathbb{R}^d$ a general mechanism to answer the query privately is to compute $f(D)$ and then add noise. The global sensitivity of a function was introduced in the celebrated paper by Dwork, McSherry, Nissim and Smith~\cite{Dwork_McSherry_Nissim_Smith_2017}, who showed that it suffices to perturb the output of the function with noise proportional to the global sensitivity of the function in order to preserve differential privacy.

\begin{definition}[Global sensitivity]\label{def:GS}
For a query $f: \mathcal{D} \to \mathbb{R}^d$, the global sensitivity of $f$ (wrt the $\ell_1$-metric) is given by 
$$ GS_f = \max_{A,B \in \mathcal{D} : A \sim B} \| f(A) - f(B) \|_1 \;. $$
\end{definition}

One can preserve differential privacy by computing $f(D)$ and adding Laplacian noise\footnote{ Here, the probability density function of the Laplace distribution $\Lap(\lambda)$ is $h(z)= \frac{1}{2\lambda} \exp\left( -\frac{\vert z \vert }{\lambda}\right)$.} scaled to the global sensitivity of $f$, where $D$ is a database. However, in many contexts we may not be able to compute the function $f$ exactly. For example, if the dataset $D$ is very large and our algorithm needs to run in sublinear-time or if the function $f$ is intractable e.g., $f(G)$ is the size of the minimum vertex cover. In cases where we cannot compute $f$ exactly, an attractive alternative is to use a randomized algorithm, say $\cA_f$,  to approximate the value of $f$. Given an approximation algorithm $\cA_f$ it is natural to ask whether or not we can add noise to $\cA_f(D)$ to obtain a differentially private approximation of $f(D)$ and (if possible) how to scale the noise. We first observe that computing $\cA_f(D)$ and adding noise scaled to the global sensitivity of $f$ does not necessarily work. Intuitively, this is because the sensitivity of $\cA_f$ can be vastly different from that of $f$. For example, suppose that $GS_f=1$, $f(D) = n = f(D')+1$ for neighboring datasets $D\sim D'$ and that our approximation algorithm guarantees that $0.999 \cdot f(D) \leq \mathcal{A}_f(D) \leq  1.001 \cdot f(D)$. It is possible that $A_f(D) = 1.001 n$ and $A_f(D')=0.999(n-1)$ so that $|A_f(D)-A_f(D')| \geq 0.002 n$ which can be arbitrarily larger than $GS_f$ as $n$ increases.
 
{\bf Coupled Global Sensitivity.} We propose the notion of {\em coupled global sensitivity} of randomized algorithms as a framework for providing general-purpose privacy mechanisms for approximation algorithms running on a database $D$. In this framework, our differentially-private algorithms can follow a unified strategy, in which in the first step a non-private randomized approximation algorithm $\cA_f(D)$ is run on the dataset, and privacy is obtained by adding Laplace noise proportional with the coupled global sensitivity of $\cA_f$\footnote{We note that this is the simplest application of CGS, and as we will see in the analysis of estimating the average degree, we can use CGS to add noise to intermediate quantities used by the randomized algorithm as well.}. The concept of coupled global sensitivity has been used implicitly in prior work on differential privacy e.g., see~\cite{alabi2020differentially,CV_NIPS2013}. Our work formalizes this notion as a general tool that can be used to design and analyze differentially private approximation algorithms. See~Appendix~\ref{sec:cgs-example} for a motivating example.

 {\bf Notation:} When $\cA$ is a randomized algorithm we use the notation $x:=\cA(D; r)$ to denote the output when running $\cA$ on input $D$ with fixed random coins $r$. Similarly, $\cA(D)$ can be viewed as a random variable taken over the selection of the random coins $r$.
\begin{definition}[Coupling]
Let $Z$ and $Z'$ be two random variables defined over the probability spaces $\cZ$ and $\cZ'$, respectively. A coupling of $Z$ and $Z'$, is a joint variable $(Z_c,Z'_c)$ taking values in the product space $(\cZ \times \cZ')$ such that $Z_c$ has the same marginal distribution as $Z$ and $Z'_c$ has the same marginal distribution as $Z'$. The set of all couplings is denoted by ${\sf Couple}(Z,Z')$.
\end{definition}

\begin{definition}[Coupled global sensitivity of a randomized algorithm]
Let $\cA: \mathcal{D} \times \mathcal{R}  \to \bbR^k$ be a randomized algorithm that outputs a real-valued vector.  Then the {\em coupled global sensitivity} of $\cA$ is defined as %\enote{isn't $\sigma: \cR \to  \cR$ instead of $\sigma: R \to  R$ ?}\tanote{changed.}
\[CGS_{\cA} := \max_{D_1 \sim D_2}\ \min_{C \in {\sf Couple}(\cA(D_1),\cA(D_2))}\ \max_{(z,z') \in C} \| z -z' \|_1 \]
\end{definition} 

\begin{remark}
We can try to relax the definition of Coupled Global Sensitivity as follows: $CGS_{\cA,
\delta}$ is the minimum value, say $x$ such that for all neighboring inputs $D_1 \sim D_2$, there exists a coupling $C$ such that $\Pr_{(z,z')\sim C}[\vert z - z' \vert > x] \leq \delta$. We need to be careful here as we need to ensure that the minimum value $x$ is always well-defined. If we can ensure this, then we can also show that adding noise proportional to $CGS_{\cA, \delta}$ preserves $(\eps,\delta)$-differential privacy. 
\end{remark}
\begin{fact}\label{fact:cgs-perm}
Let $\cA: \mathcal{D} \times \cR \to \bbR^k$ be a randomized algorithm viewed as a function that takes as input a dataset $\mathcal{D}$ and a random string in the finite set  $\mathcal{R}$, and outputs a real-valued vector. For a finite set $\cR$, denote by $Sym(\cR)$ the symmetric group of all permutations on the elements in $\cR.$ Then,
$$CGS_{\cA} \leq \max_{D_1 \sim D_2} \min_{\sigma\in Sym(\cR)} \max_{R \in \cR} \| \cA(D_1;R) -\cA(D_2; \sigma(R)) \|_1$$
\end{fact}
The following theorem formalizes the fact that adding noise proportional to the coupled global sensitivity of a randomized algorithm preserves differential privacy (see Appendix~\ref{sec:cgs} for a formal proof).
\begin{theorem}\label{thm:lap-cgs}
Let $\cA: \mathcal{D} \to \bbR^k$ be a randomized algorithm and define the Laplace mechanism 
$\cM_L(D) = \cA(D) + (Y_1, \ldots, Y_k) $, where $Y_i$ are i.i.d. random variables drawn from $\Lap(CGS_\cA/ \eps)$. The mechanism $\cM_L$ preserves $\epsilon$-differential privacy.
\end{theorem} 
{\bf How we use Coupled Global Sensitivity.} In our algorithm for estimating the average degree we divide the algorithm into randomized sub-routines and show that the CGS of these sub-routines is small, therefore enabling us to add Laplacian noise proportional to the CGS and ensure the privacy of each sub-routine, and by composition, the privacy of the entire algorithm (See Theorem~\ref{thm:main-privacy-avgdeg}). Similarly, we show that the existing non-private sublinear-time algorithms for maximum matching and minimum vertex cover have small CGS, therefore enabling us to add Laplace noise proportional to the CGS to their outputs thus making them differentially-private (See Theorems~\ref{thm:sub-MM-cgs},\ref{thm:sub-VC-cgs}).

\subsection{Technical Overview}\label{sec:tech-ov}
\subsubsection{Privately Estimating the Average Degree.}

At a high-level, our private algorithm for estimating the average degree follows the non-private variant of Goldreich and Ron~\cite{goldreich2004estimating}. However, there are several challenges that prevent us from simply being able to add Laplacian noise to the output. We overcome these challenges by first
obtaining a new non-private algorithm with the same approximation ratio as that of~\cite{goldreich2004estimating}, and then further add appropriate amounts of noise in several steps of the algorithm to obtain both privacy and accuracy guarantees. We begin by describing the algorithm of~\cite{goldreich2004estimating}.

{\bf The Goldreich-Ron algorithm~\cite{goldreich2004estimating}.} The strategy of the original non-private algorithm in~\cite{goldreich2004estimating} is to sample a set $S$ of vertices partition them into buckets $S_i$ based on their degrees. In particular, for each $i$ we set $S_i = B_i \cap S$ where the set $B_i$ contains all vertices of degrees ranging between $((1+\beta)^{i-1}, (1+\beta)^{i}]$, where $\beta=\rho/c$ for some constant $c>1$. Intuitively, as long as $|S_i|$ is sufficiently large the quantity $\vert S_i \vert/\vert S \vert$ is a good approximation for $\vert B_i \vert /n$ with high probability. Let $I$ denote the indices $i$ for which $|S_i|$ is sufficiently large. We can partition edges from the graph into three sets (1) edges with both endpoints in $\bigcup_{i \in I} B_i$, (2) edges with exactly one endpoint in $\bigcup_{i \in I} B_i$, and (3) edges with no endpoints in $\bigcup_{i \in I} B_i$. When the threshold for "large buckets" is  tuned appropriately one can show that (whp) type 3 edges can be ignored as there are at most $o(n)$ such edges. 

We could use $(1/\vert S \vert) \sum_{i \in I} \vert S_i \vert (1+\beta)^{i-1}$ as an approximation for $\frac{1}{n} \sum_{i \in I} \sum_{v \in B_i} \mathtt{deg}(v)$. The previous sum counts type (1) edges twice, type (2) edges once and type (3) edges zero times. While it is ok to ignore type (3) edges there could be a lot of type (2) edges which are under-counted. To correct for type (2) edges we can instead try to produce an approximation for the sum $\frac{1}{n} \sum_{i \in I} \sum_{v \in B_i} (1+\alpha_v) \mathtt{deg}(v)$ where $\alpha_v$ denotes the fraction of type (2) edges incident to $v$. Intuitively, $\alpha_v$ is included to ensure that type (2) edges are also counted twice. For each sampled node $v \in S_i$ we can pick a random neighbor $r(v)$ of $v$ and define $X(v) = 1$ if $r(v) \not \in \bigcup_{i \in I} B_i$; otherwise $X(v)=0$. Observe that in the expected value of the random variable is $\mathbb{E}[X(v)] = \alpha_v$. Since $|S_i|$ is reasonably large for each $i \in I$ and $\mathtt{deg}(u) \approx \mathtt{deg}(v)$ for each pair $u,v \in S_i$ we can approximate the fraction of type (2) edges incident to $B_i$ as $W_i/|S_i|$ where $W_i =  \sum_{v \in S_i} X(v)$. Finally, we can use $(1/\vert S \vert) \sum_{i \in I} |S_i| (1+W_i/|S_i|)  (1+\beta)^{i-1}$ as our final approximation for the average degree.

{\bf Challenges to making the original algorithm private by adding noise naively. } The first naive attempt to transform the algorithm of \cite{goldreich2004estimating} into a differentially private approximation would be to add noise to the final output. However, the coupled global sensitivity of this algorithm is large enough that the resulting algorithm is no longer a $(1+\rho)$-approximation.

A second natural strategy to make the above algorithm differentially private is to add Laplace noise to the degree of each vertex and partition vertices in $S$ based on their noisy degrees $\tilde{d}(v) = \mathtt{deg}(v)+Y_v$ where $Y_v \sim \mathtt{Lap}(6/\eps)$. (Note: To ensure that the algorithm still runs in sublinear time we could  utilize lazy sampling and only sample $Y_v \sim \mathtt{Lap}(6/\eps)$ when needed). In particular, we can let $\tilde{S}_i = S \cap \tilde{B}_i$ where $\tilde{B}_i$ denotes the set of all nodes $v$ with noisy degree $\tilde{d}(v)$ ranging between $((1+\beta)^{i-1}, (1+\beta)^{i}]$. Now we can compute $W_i = Z_i +  \sum_{v \in \tilde{S}_i} X(v)$ where $Z_i \sim \mathtt{Lap}(6/\epsilon)$ and return $(1/\vert S \vert) \sum_{i \in I} |\tilde{S}_i| (1+\frac{W_i}{|\tilde{S}_i|})  (1+\beta)^{i-1}$. While the above approach would preserve differential privacy, the final output may not be accurate. The problem is that the noise $Y_v$ may cause a node $v$ to shift buckets. It is not a problem if $v \in B_i$ shifts to an adjacent bucket i.e., $v \in \tilde{B}_{i-1}$ or $v \in \tilde{B}_{i+1}$ since $(1-\beta)^{i-2}$ and $(1-\beta)^{i+1}$ are still reasonable approximations for the original degree $\deg(v) \in ((1+\beta)^{i-1},(1+\beta)^i]$. Indeed, when $\deg(v)$ is sufficiently large we can argue that $(1-\beta) \deg(v) < \tilde{d}(v) < (1+\beta) \deg(v)$ with high probability. However, this guarantee does not apply when $\deg(v)$ is small. In this case the Laplace noise $Y_v$ might dominate $\deg(v)$ yielding an inaccurate approximation. 
 \cite{sivasubramaniamdifferentially} made similar observations, and because of these technical barriers, their paper analyzes the simpler strategy for estimating the average degree, which yields a less accurate result. 
The crucial observation here is that we need to deal with vertices having small degrees in our accuracy analysis separately.

{\bf Modified non-DP algorithm achieving the same approximation ratio.} To address the challenges discussed above we first propose a modification to the strategy given by~\cite{goldreich2004estimating}. While the modified algorithm is still non-private it still achieves a $(1+\rho)$-approximation for any $\rho >0$ and is amenable to  differentially private adaptations. Our algorithm now samples vertices $S$ \emph{without replacement} and puts them into buckets $ S_i = B_i \cap S$ according to their degrees. The key modification is that we merge all of the buckets with smaller degrees i.e., $i \leq K$\footnote{where we fix $K:=\left(2+ \log_{1+\beta} \left(\frac{2 |S| \sqrt{\rho} }{\beta \log_{1+\beta}(n) \sqrt{n \sqrt{\log n}} } \right)  \right)$ in the sequel} into one. We redefine $B_1$ to denote this merged bucket and $S_1 = S \cap B_1$ and we redefine $I$ to be the set of all indices $i > K$ such that $|S_i|$ is sufficiently large. If $B_1$ is not too large then all of the edges incident to $B_1$ can simply be ignored as the total number of these edges will be small. Otherwise, we can account for edges that are incident to $B_1$ by adding $\frac{1}{|S|} \sum_{v \in S_1} (1+X(v)) \deg(v)$ to our final output. Since we merged all of the buckets with smaller degrees we no longer have the guarantee that $\deg(u) \approx \deg(v)$ for all $u, v \in S_1$. However, since $\deg(v)$ is reasonably small for each $v \in S_1$ the variance is still manageable. Intuitively, the sum $\frac{1}{|S|} \sum_{v \in S_1} (1+X(v)) \deg(v)$ approximates $ \frac{1}{n} \sum_{v \in B_1} (1+\alpha_v) \deg(v)$ where $\alpha_v$ now denotes the fraction of edges incident to $v$ whose second endpoint lies outside  the set $B_1 \cup \bigcup_{i \in I} B_i$.

{\bf The differentially-private modified algorithm.}\label{sec:dp-avg-overview} We now introduce our sublinear-time differentially-private algorithm to approximate the average degree in Algorithm~\ref{alg:1+e}. Algorithm~\ref{alg:1+e} relies on three subroutines given by Algorithms~\ref{alg:noisy-deg}, \ref{alg:bigsmall}, and \ref{alg:noisy-avg-deg}. Splitting the algorithm into separate modules simplifies the privacy analysis as we can show that each subroutine is $\epsilon/3$-differentially private --- it follows that the entire algorithm is $\epsilon$-differentially private. In Algorithm~\ref{alg:noisy-deg} we add Laplace noise to the degrees of \emph{all} vertices in the graph and then return a sample of vertices, say $S$ (sampled uniformly without replacement) along with their noisy degrees. For simplicity we describe Algorithm~\ref{alg:noisy-deg} in a way that the running time is linear in the size of the input. We do this to make our privacy analysis simpler. However, we can implement Algorithm~\ref{alg:noisy-deg} with lazy sampling of Laplace noise $Y_u$ when required i.e., if node $u$ is in our sample $S$ or if $u=r(v)$ was the randomly selected neighbor of some node $v \in S$.
\begin{center}
\vspace{-0.3cm}
\fbox{\parbox{\textwidth}{
{\textbf{NoisyDegree} takes $G$ as input and returns a set of sampled vertices along with the noisy degrees of every vertex in $G$.}  
      \begin{enumerate}[nolistsep]
      \item Uniformly and independently select $\Theta(\sqrt{n} \cdot \text{poly}(\log (n)/\rho) \cdot \text{poly}(1/\eps))$ vertices (without replacement) from $V$ and let $S$ denote the set of selected vertices.
      \item For every $v \in V(G)$, $$\tilde{d}(v)=\deg(v)+Y_v \;,$$ where $Y_v \sim \Lap(6/\eps)$.
      \item Return $\{\tilde{d}(v)\}_{v \in V(G)}, S$
    \end{enumerate}
}}
\captionof{algorithm}{\textbf{NoisyDegree}}  
\label{alg:noisy-deg}
\vspace{-0.3cm}
\end{center}

Given the output of Algorithm \ref{alg:noisy-deg} we can partition the sample $S$ into buckets $\tilde{S}_i = S \cap \tilde{B}_i$ using their noisy degree. Here, we define $\tilde{B}_i = \left\{ v : \tilde{d}(v) \in \left((1+\beta)^{i-1}, (1+\beta)^i\right)\right\}$ and we also define a merged bucket $S_1 = S \cap \left\{ v : \tilde{d}(v) \leq (1+\beta)^{K-1} \right\}$ containing all sampled nodes with noisy degree at most $ (1+\beta)^{K-1}$. Here, $K$ is a degree threshold parameter that we can tune. Now given a size threshold parameter $T$ we can define $I=\left\{ i \geq K ~:~\left|\tilde{S}_i \right| \geq 1.2 T \cdot |S| \right\}$ to be the set of big buckets. We remark that as a special case we define $|S_1|$ to be ``small" if $|S_1| < 1.2 T \cdot \sqrt{|S|} \cdot |S|$ instead of $|S_1| < 1.2 T \cdot|S|$. As an intuitive justification we note that (whp) for each node $v$ with noisy degree $\tilde{d}(v) \leq (1+\beta)^{K-1}$ the actual degree $\deg(v)$ will not be too much larger than $(1+\beta)^{K-1}$. In this case we have  $\sum_{v: \tilde{d}(v) \leq (1+\beta)^{K-1}} \deg(v) \leq |S_1| \max_{v: \tilde{d}(v) \leq (1+\beta)^{K-1}} \deg(v) = o(n)$ so that we can safely ignore the edges incident to $S_1$.

Intuitively, for each large bucket $i \in I$, Algorithm~\ref{alg:bigsmall} computes $\tilde{\alpha}_i= W_i/|\tilde{S}_i|$ our approximation of the fraction of type (2) edges incident to $\tilde{B}_i$. If $S_1$ is large then type (2) edges are (re)defined to be the edges with exactly one endpoint in $\left\{ v : \tilde{d}(v) \leq (1+\beta)^{K-1} \right\} \cup \bigcup_{i \in I} \tilde{B}_i$. To preserve differential privacy we add laplace noise to $W_i$ i.e., $W_i = Z_i +\sum_{v \in \tilde{S}_i} X(v)$ where $Z_i \sim \mathtt{Lap}(6/\epsilon)$. We remark that (whp) we will have $Z_i = o(|\tilde{S}_i|)$ for each large bucket $i \in I$. Thus, the addition of laplace noise will have a minimal impact on the accuracy of the final result.

\begin{center}
\vspace{-0.3cm}
\fbox{\parbox{\textwidth}{
\textbf{NoisyBigSmallEdgeCount} takes as input $G,I,\{\tilde{S}_i\}^t_{i=1}, S_1, \{\tilde{d}(v)\}_{v \in V(G)}, M_{\rho,n}, T$ and returns the fraction of edges that are between big buckets and not big buckets.
      \begin{enumerate}[nolistsep]
       \item For every $i \in I$, count the edges between buckets in $I$ and small buckets,
    \begin{enumerate}
        \item For all $v \in \tilde{S}_i$, 
    \begin{enumerate}
        \item Pick a random neighbor of $v$, say $r(v)$. 
        \item If $\vert S_1 \vert <1.2 T \cdot \sqrt{\vert S \vert }\cdot \vert S \vert $, i.e., if $S_1$ is a small bucket. Then if $\tilde{d}(r(v))\in ((1+\beta)^{i-1},(1+\beta)^i]$ for some $i \not\in I$, then $X(v)=1$, otherwise $X(v)=0$. 
        \item Otherwise, $S_1$ is not small. Therefore, if $\tilde{d}(r(v)) \in ((1+\beta)^{i-1},(1+\beta)^i]$ for some $i \not\in I$ and $i > \log_{1+\beta} \lceil \left(\frac{6M_{\rho,n}}{\beta}\right) \rceil+2$, then $X(v)=1$, otherwise $X(v)=0$. 
    \end{enumerate}
\item Define $ W_i := \sum_{v \in \tilde{S}_i} X(v) + Z_i$  where $Z_i \sim \Lap(6/\eps)$ and $\tilde{\alpha}_i := \frac{W_i}{\vert \tilde{S}_i \vert }$.    
    \end{enumerate}
   \item return $\{W_i\}_{i\in I},\{\tilde{\alpha}_i\}_{i\in I} $
    \end{enumerate}
}}
\captionof{algorithm}{\textbf{NoisyBigSmallEdgeCount}}  
\label{alg:bigsmall}
\vspace{-0.3cm}
\end{center}

If the merged bucket $S_1$ is small then we can ignore edges incident to $S_1$ and Algorithm~\ref{alg:noisy-avg-deg} will simply output $\frac{1}{|S|} \sum_{i \in I} |\tilde{S}_i| \cdot (1+\tilde{\alpha}_i) \cdot (1+\beta)^i$. In this case the output can be computed entirely from the differentially private outputs that have already been computed by Algorithms \ref{alg:noisy-deg} and \ref{alg:bigsmall} without even looking at the graph $G$. Intuitively,  for any large bucket $i \in I$ and $v \in \tilde{S}_i$ we expect that (whp) $|Y_v| = |\tilde{d}(v) - \deg(v)|$ is small enough to ensure that $(1+\beta)^{i-2} \leq \deg(v) \leq (1+\beta)^{i+1} $. Thus, $(1+\beta)^i$ is still a reasonable approximation for $\deg(v)$. 

If the merged bucket $S_1$ is sufficiently large, then we need to account for the edges within $S_1$ itself as well as the fraction of edges between $S_1$ and small buckets. We introduce a new estimator to approximate the fraction of edges between $S_1$ and small buckets given by $Z+\sum_{v \in S_1} (1+X(v)) \cdot \deg'(v)$ where $Z \sim \Lap\left(36M_{\rho,n} \left(3+\beta+\frac{1}{\beta}\right)\right)$ and $\deg'(v)= \min \{\deg(v),6M_{\rho,n} \left(3+ \beta+\frac{1}{\beta}\right)\}$ (See~Algorithm~\ref{alg:noisy-avg-deg}) --- the relationship between the parameters $K$ and  $M_{\rho,n}$ is $K = 2 + \log_{1+\beta} \lceil 6 M_{\rho, n}/\beta \rceil$. The Laplace Noise term is added to preserve differential privacy. We define the clamped degrees $\deg'(v)$ to ensure that the coupled global sensitivity of the randomized subroutine computing $\sum_{v \in S_1} (1+X(v)) \cdot \deg'(v)$ is upper bounded by $12M_{\rho,n} \left(3+ \beta+\frac{1}{\beta}\right)$. This way we can control the laplace noise parameters to ensure that $Z = o(|S_1|)$ with high probability so that the noise term $Z$ does not adversely impact accuracy. Intuitively, we expect that $Y_v \leq  M_{\rho, n}$ for all nodes $v$ with high probability. In this case for any node $v \in S_1$ we will have $\deg'(v) = \deg(v) \leq 6M_{\rho,n} \left(3+ \beta+\frac{1}{\beta}\right)$.

\begin{center}
\vspace{-0.3cm}
\fbox{\parbox{\textwidth}{
\textbf{NoisyAvgDegree} takes $\{\tilde{S}_i\}^t_{i=1}, \{\tilde{d}(v)\}_{v \in V(G)},\{\tilde{\alpha}_i\}_{i\in I},I, M_{\rho,n}, T$ as input and returns the noisy estimator for average degree of the graph.
      \begin{enumerate}[nolistsep]
        \item \label{it:case1-output} If $\vert S_1 \vert < 1.2T \cdot \sqrt{\vert S \vert}\cdot \vert S \vert $ then output $\frac{1}{\vert S \vert } \sum_{i \in I} \vert \tilde{S}_i \vert \cdot (1+\tilde{\alpha}_i) \cdot (1+\beta)^{i}$.  
 \item Else, $S_1$ is a big bucket, and we need to count edges between $S_1$ and small buckets. Thus, for every $v \in S_1$, 
 \begin{enumerate}
     \item Pick a random neighbor of $v$, say $r(v)$.
     \item If $\tilde{d}(r(v)) \in ((1+\beta)^{i-1},(1+\beta)^i]$ for some $i \not\in I$ and $i > \log_{1+\beta} \lceil \left(\frac{6M_{\rho,n}}{\beta}\right) \rceil+2$, then $X(v)=1$, otherwise $X(v)=0$. 
     \item  \label{it:case2-output} Output $$\frac{1}{\vert S \vert } \left(\sum_{i \in I } \vert \tilde{S}_i \vert \cdot (1+\tilde{\alpha}_i) \cdot (1+\beta)^{i} +Z+\sum_{v \in S_1} (1+X(v)) \cdot \deg'(v)  \right) \;,$$ where $Z \sim \Lap\left(36M_{\rho,n} \left(3+ \beta+\frac{1}{\beta}\right)\right)$ and $\deg'(v)= \min \{\deg(v),6M_{\rho,n} \left(3+ \beta+\frac{1}{\beta}\right)\}$.  
 \end{enumerate}
    \end{enumerate}
}}
\captionof{algorithm}{\textbf{NoisyAvgDegree} } 
\label{alg:noisy-avg-deg}
\vspace{-0.3cm}
\end{center}

\begin{center}
\fbox{\parbox{\textwidth}{
\textbf{Main DP Algorithm} that takes graph $G$ as input and outputs an approximation of its average degree.
      \begin{enumerate}[nolistsep]
        \item $\{\tilde{d}(v)\}_{v \in V(G)}, S:= \textbf{NoisyDegree}(G)$ \Comment{see Algorithm \ref{alg:noisy-deg}}
\item For {$i = 1,2\ldots, t$},  let $\tilde{S}_i= \{ v \in S :\ \tilde{d}(v) \in ((1+\beta)^{i-1},(1+\beta)^i] \}$ where $t:= \lceil \log_{(1+\beta)} (n) \rceil$.
\item Define $M_{\rho,n}:= \frac{1}{3} \cdot \sqrt{\frac{\rho}{n \sqrt{\log (n)}}} \cdot \frac{\vert S \vert }{t}$, $S_1:= \cup_{i\leq \log_{1+\beta} \left(\frac{6M_{\rho,n}}{\beta}\right)+2}  \tilde{S}_i$, and, $I=\{i > \log_{1+\beta} \left(\frac{6M_{\rho,n}}{\beta}\right)+2\ :\ \vert \tilde{S}_i \vert \geq 1.2T \cdot \vert S \vert \}$ where $T:= \frac{1}{2}\sqrt{\frac{\rho}{n}}\cdot \frac{\eps}{(1+\eps)} \cdot \frac{1}{t}$. % where $Y_{i} \sim Lap(2/\eps)$. 
\item $\{W_i\}_{i\in I},\{\tilde{\alpha}_i\}_{i\in I} := \textbf{NoisyBigSmallEdgeCount}(G,I,\{\tilde{S}_i\}^t_{i=1})$. \Comment{see Algorithm \ref{alg:bigsmall}}
\item \textbf{NoisyAvgDegree}($G,S,\{\tilde{S}_i\}^t_{i=1}, \{\tilde{\alpha}_i\}_{i\in I},I, M_{\rho,n},T$).  \Comment{see Algorithm \ref{alg:noisy-avg-deg}}
    \end{enumerate}
}}
\captionof{algorithm}{\textbf{Main DP Algorithm}}  
\label{alg:1+e}
\end{center}
The full analysis of Theorem~\ref{thm:informal-avg-deg} can be found in Sections~\ref{sec:privacy-avgdeg} and~\ref{sec:accuracy-avgdeg}.

\subsubsection{Privately Estimating Maximum Matching and Vertex Cover Size}\label{sec:tech-ov-match}

At a high-level our private algorithms for estimating the maximum matching and vertex cover add laplace noise (to the outputs) proportional to the coupled global sensitivity of the randomized non-private algorithms for the same. The challenge lies in proving the coupled global sensitivity of these non-private algorithms is small. 

We first describe and analyze the coupled global sensitivity of the classical polynomial-time greedy matching algorithm. This is helpful in our analysis of the non-private sublinear-time algorithm for maximum matching in the sequel. 

We then describe and give a proof sketch of the coupled global sensitivity of the non-private sublinear-time matching algorithm~\cite{behnezhad21}. The full proofs for privately estimating the maximum matching and minimum vertex cover size are in Section~\ref{sec:main-sub-mm}. 
Recall that $d$ is the maximum degree and $\bar{d}$ is the average degree of the graph .

{\bf The Polynomial-time Greedy Matching Algorithm $\cA_{MM}$.} This algorithm takes as input a graph $G=(V, E)$ and a random permutation $\pi$ on the set of pairs $(x, y)\in V\times V$, with $x\ne y$, and processes each pair of vertices $(x,y)$ in the increasing order of ranks given by $\pi$, and greedily adds edges to a maximal matching whose size is finally output\footnote{We note that the non-private algorithms~\cite{nguyen2008constant,yoshida2012improved,ORR12} only consider the ranking $\pi$ over $m$ edges of the graph, whereas we consider the ranking over all $\binom{n}{2}$ pairs of vertices. This is because we want to define a ``global'' ranking so that we can define the same ranking consistently over neighboring graphs that may have different edges.}. Since the size of the maximal matching produced is known to be at least $\frac12$ of the size of a maximum matching, this gives a non-private $2$-approximation of the size of a maximum matching in $G$.

{\bf CGS of the Greedy Algorithm $\cA_{MM}$.}
We show that the CGS of the greedy algorithm (with respect to node-neighboring graphs) is at most 1. Note that once the ranking on the edges is fixed the maximal matching obtained by $\cA_{MM}$ is also fixed. 
Let $\sigma_I$ be the identity permutation over the ranking of edges, i.e., we have $\sigma_I(\pi)=\pi$. We use Fact~\ref{fact:cgs-perm} to observe that,
\begin{align*}
CGS_{\cA_{MM}}&\leq   \max_{G_1 \sim G_2} \min_\sigma \max_{\pi} \vert \cA_{MM}(G_1;\pi)-\cA_{MM}(G_2;\sigma(\pi)) \vert \\
&\leq \max_{\substack{G_1 \sim G_2 \\ \pi}} \vert \cA_{MM}(G_1;\pi)-\cA_{MM}(G_2;\sigma_I(\pi)) \vert\\
&=\max_{\substack{G_1 \sim G_2 \\ \pi}} \vert 
\cA_{MM}(G_1;\pi)-\cA_{MM}(G_2;\pi).\vert
\end{align*}
Therefore it is sufficient to analyze the relative size of the matching obtained on node-neighboring graphs $G_1,G_2$ that are processed by the greedy algorithm in the order given by the same $\pi$. 

Let $G_1 \sim G_2$ where $v^*$ is such that $E(V_1\setminus \{v^*\}) = E(V_2\setminus \{v^*\})$. 
Denote the greedy matchings obtained from $\cA_{MM}(G_1,\pi)$ as $M_1$ and from $\cA_{MM}(G_2,\pi)$ as $M_2$. Suppose edge $e^*$ is incident to $v^*$ such that $e^* \in E_2$, and $e^* \not\in E_1$. We will show that  $\vert \vert M_1 \vert - \vert M_2 \vert \vert \leq 1$, which implies that $\max_{\substack{G_1 \sim G_2 \\ \pi}} \vert \cA_{MM}(G_1;\pi)-\cA_{MM}(G_2;\pi) \vert \leq 1$, thus proving that $CGS_{\cA_{MM}} \leq 1$.

We first claim that if  $e^* \not\in M_1 \cup M_2$ then $|M_1|=|M_2|$. Since the greedy algorithm considers edges in the same order, the exact same edges must have been placed in  $M_1$ as in $M_2$ before $e*$ is processed. Since $e^*=(v^*, u)$ is not chosen in $M_2$ it must have been the case that by this time $u$ was matched in $M_2$, and thus the same matched edge must occur in $M_1.$ From here on the algorithm again must make the same choices for the edges to be placed in $M_1$ and $M_2$.

Next, we claim that if $e^*\in M_1\cup M_2$ then $M:= M_1 \oplus M_2$\footnote{$M_1 \oplus M_2$ is the \emph{symmetric difference} of sets and this is defined as the set of edges in either $M_1$ or $M_2$ but not in their intersection.} is one connected component containing $e^*$.  Consequently, $\vert \vert M_1 \vert - \vert M_2 \vert \vert \leq 1$. Since $e^* \in M_2$ and $e^*$ cannot be in $M_1$, it is clear that $e^* \in M$. Suppose for the sake of contradiction, $M$ consists of two connected components $C_1,C_2$ and WLOG $e^* \in C_1$.  Consider edges in $C_2$. By Berge's Lemma \cite{west2001introduction}, $C_2$ is either an alternating path or an alternating even cycle, with alternating edges from $M_1$ and $M_2$. Also, the edges in $C_2$ exist in both $G_1$ and $G_2$ with the same ranking. Observe that since $C_2$ is separate from $C_1$ containing $e^*$, if we replace edges in $C_2$ belonging to $M_2$ in the original graph $G_2$ by edges in $C_2$ belonging to $M_1$, this is still a valid maximal matching for the graph $G_2$. In fact, the greedy algorithm considers edges in $C_2$ in the same order for both graphs $G_1,G_2$, so the edges in $M_1$ and $M_2$ should be the same, in other words, $C_2$ cannot be a part of $M=M_1 \oplus M_2$, and hence $M$ must have only one connected component, which contains $e^*$. Now, since $M$ is either an alternating path or even cycle, $\vert \vert M_1 \vert - \vert M_2 \vert \vert \leq 1$.

{\bf The Local Maximum Matching Algorithm $\cA_{sub-MM}$.}

We describe the local algorithm implemented by~\cite{behnezhad21} in Algorithm~\ref{alg:MM-SLA-sampling}. We modify the original algorithm to sample vertices without replacement. The algorithm then calls the vertex cover oracle (denoted as $\mathcal{O}^\pi_{VC}$) on each sampled vertex which subequently calls the maximal matching oracle (denoted as $\mathcal{O}^\pi_{MO}$) on the incident edges to determine whether the sampled vertex is in the matching fixed by the ranking of edges $\pi$. Finally, the algorithm returns an estimate of the maximum matching size based on the number of sampled vertices in the matching. We note that in~\cite{behnezhad21} the same sampling algorithm simultaneously outputs an approximation to maximum matching size and minimum vertex cover size. We choose to write the sampling procedure for estimating the maximum matching size (see Algorithm~\ref{alg:MM-SLA-sampling}) and minimum vertex cover size (see Algorithm~\ref{alg:VC-SLA-sampling}) separately so that it is easier to understand the Coupled Global Sensitivity for outputting the two different estimators.

\begin{center}
\fbox{\parbox{\textwidth}{
  \textbf{Input. }Input Graph $G=(V,E)$.
\begin{enumerate}[nolistsep]
    \item Uniformly sample $s=16 \cdot 24 (\ln n)/\rho^2$ vertices from $V$ without replacement.
    \item For{ $i=1 \ldots s$,} if {$\mathcal{O}^\pi_{VC}(v_i)=$ True} then let $X_i = 1$, otherwise let $X_i=0$. 
    \item return $\tilde{M}= \frac{n}{2s}(\sum_{i \in [s]}X_i)-\frac{\rho n}{2}$.
    \end{enumerate} }}
    \captionof{algorithm}{Local Maximum Matching algorithm $\mathcal{A}_{sub-MM}$ using Oracle access.}  
\label{alg:MM-SLA-sampling}
\end{center}

\begin{remark}
\cite{behnezhad21}~gives an efficient simulation of the matching and vertex cover oracles which exposes edges incident to a vertex in batches only when they are needed. We assume the efficient simulation of these oracles in our algorithms.  
\end{remark}

{\bf CGS of the Local Matching Algorithm $\cA_{sub-MM}$.} We first describe the challenges to analyzing the coupled global sensitivity of $\cA_{sub-MM}$ and then give an intuition for why the $CGS_{\cA_{sub-MM}} \leq \frac{n\rho^2}{16 \cdot 24 \ln n}$. 

A naive approach could be to consider the identity permutation $\sigma_I$ over the ranking of edges $\pi$ and sampled vertices $v\in V $. But this approach does not work. For node-neighboring graphs $G_1, G_2$, it could be the case that all the edges sampled from $G_1$ belong to the matching $M_1$ fixed by the ranking $\pi$, but the same edges sampled from $G_2$ may not be in the matching $M_2$ fixed by the ranking $\pi$. Thus, we need to carefully define a bijection that maps endpoints of edges in the matching $M_1$ to endpoints of edges in the matching $M_2$. By using this bijection to define a permutation and the fact that $\vert M_2 \setminus M_1 \vert \leq 1$ which leads to the size of the sets of matched vertices in $M_1$ and $M_2$ to differ by 2, we have $CGS_{\cA_{sub-MM}} \leq \frac{n}{s} \leq  \frac{n\rho^2}{16 \cdot 24 \ln n}$.

\section{Privacy Analysis of Theorem~\ref{thm:informal-avg-deg}}\label{sec:privacy-avgdeg}%

\begin{theorem}\label{thm:main-privacy-avgdeg}
 The Algorithm~\ref{alg:1+e} is $\eps$-DP. 
 \end{theorem}
 \begin{proof}
 We will approach the privacy analysis in a modular fashion, i.e., we will analyze each sub-routine separately and show that by composition, the entire algorithm is $\eps$-differentially private.

In the sequel, when analyzing the coupled global sensitivity of intermediate randomized quantities, we use Fact~\ref{fact:cgs-perm}. 
 
\begin{claim}
Algorithm~NoisyDegree (see Algorithm~\ref{alg:noisy-deg}) is $\eps/3$-DP.
\end{claim}   
\begin{proof}
First, fix any sample $S$. Define the function $f_{noisy-deg} := \{\tilde{d}(v)\}_{v\in V(G)}$. Observe that the degree of a node can change by at most 1 from adding or deleting an edge, and therefore $f_{noisy-deg}$ changes by at most 2 by adding or deleting an edge, in other words, the $GS_{f_{noisy-deg}}=2$ and we can add noise proportional to $2/\eps$.  
\end{proof}
\begin{claim}
Algorithm~NoisyBigSmallEdgeCount (Algorithm~\ref{alg:bigsmall}) is $\eps/3$-DP. 
\end{claim}
\begin{proof}
 We fix noisy degrees $\{\tilde{d}(v)\}_{v \in V(G)}$, consequently fixing the buckets $\tilde{S}_1,\ldots, \tilde{S}_t$ and set $I$. Define the function $f_{t,\tilde{d}}:= \{f_{\tilde{S}_i, \tilde{d}}(G;r) \}_{i\in I}$, and the function $f_{\tilde{S}_i, \tilde{d}}(G;r) = \sum_{v \in \tilde{S}_i} H(r(v))$ where $H(w) = 1$ if and only if we have $\tilde{d}(w) \in ((1+\beta)^{i-1}, (1+\beta)^i]$ for some $i  \not \in I$ and $\vert S_1 \vert <1.2 T \cdot \sqrt{\vert S \vert} \cdot \vert S \vert $ or if $\tilde{d}(w) \in ((1+\beta)^{i-1}, (1+\beta)^i]$ for some $i  \not \in I$ and $i > \log_{1+\beta} \lceil \left(\frac{6M_{\rho,n}}{\beta}\right) \rceil+2$; here $r(\cdot)$ defines the random coins used to sample a neighbor of $v$. We analyze $CGS_{f_{\tilde{S}_i, \tilde{d}}}$, and argue that $CGS_{f_{t,\tilde{d}}} \leq CGS_{f_{\tilde{S}_i, \tilde{d}}}$.

 First, we show that for all fixed $S$, $\{\tilde{d}(v)\}_{v \in S}$ and $i\in I$, the $CGS_{f_{S_i,\tilde{d}}}$ is at most 2. Consider $G$ and $G’$ such that edge $(u^*,v^*) \in G$, but does not exist in $G'$. Fix any coupling such that $r(w) = r’(w)$ for all $w \neq {u^*,v^*}$, where $r,r'$ defines the random coins for sampling neighbors of $w$ in $G$ and $G'$ respectively. Now we have $X(w)=H(r(w)) =H(r’(w))= X’(w)$ for all $w \neq u^*, v^*$. Thus, $ CGS_{f_{\tilde{S}_i,\tilde{d}}} = \vert f_{\tilde{S}_i,\tilde{d}}(G;r) - f_{\tilde{S}_i,\tilde{d}}(G’;r') \vert  = \vert \sum_{v \in \tilde{S}_i} H(r(v))-  \sum_{v \in \tilde{S}_i} H(r’(v)) \vert = \vert H(r(v^*)) + H(r(u^*))- H(r’(v^*))- H(r’(u^*)) \vert \leq 2.$ Now, since the differing endpoints $u^*,v^*$ can only appear in at most one of the $i$-th iterations simultaneously, it is clear to see that $CGS_{f_{t,\tilde{d}}}$ is also at most 2.  
\end{proof}

\begin{claim}
Algorithm NoisyAvgDegree (Algorithm~\ref{alg:noisy-avg-deg}) is $\eps/3$-DP.
\end{claim}  
\begin{proof}
We fix noisy degrees $\{\tilde{d}(v)\}_{v \in V(G)}$, and sample $S$ consequently fixing the buckets $\tilde{S}_1,\ldots, \tilde{S}_t$ and set $I$, and we fix $\{ \tilde{\alpha}_i\}^t_{i=1}$.
Note that the first output in Line~\ref{it:case1-output} given by $\frac{1}{\vert S \vert } \sum_{i \in I} \vert \tilde{S}_i \vert \cdot (1+\tilde{\alpha}_i) \cdot (1+\beta)^{i}$ is already private since the terms in the summation consist of parameters that are either noisy or public or both. We need to show that the second output in Line~\ref{it:case2-output} is private. In particular, define the function $f_{S_1, \tilde{d}}(G;r):= \sum_{v \in S_1} (1+H_1(r(v))) \cdot \deg'(v)$ where $\deg'(v) = \min \{\deg(v), 6M_{\rho,n} \left(3+ \beta+\frac{1}{\beta}\right)\}$ and $H_1(w)=1$ if and only if $\tilde{d}(w) \in ((1+\beta)^{i-1},(1+\beta)^i]$ for some $i \not\in I$ and $i > \log_{1+\beta} \lceil \left(\frac{6M_{\rho,n}}{\beta}\right) \rceil+2$. We claim that for all fixed $S$ and $\{\tilde{d}(v)\}_{v \in S}$, the $CGS_{f_{S_1,\tilde{d}}}$ is at most $ 12M_{\rho,n} \left(3+ \beta+\frac{1}{\beta}\right)$. Consider $G$ and $G’$ such that edge $(u^*,v^*) \in G$, but does not exist in $G'$. Fix any coupling such that $r(w) = r’(w)$ for all $w \neq {u^*,v^*}$, where $r,r'$ defines the random coins for sampling neighbors of $w$ in $G$ and $G'$ respectively. Now we have $X(w)=H_1(r(w)) =H_1(r’(w))= X’(w)$ for all $w \neq u^*, v^*$. Thus, $ \vert f_{S_1,\tilde{d}}(G;r) - f_{S_1,\tilde{d}}(G’;r') \vert  = \vert \sum_{v \in \tilde{S}_1} (1+H_1(r(v)))\cdot \deg'(v)-  \sum_{v \in \tilde{S}_1}(1+H_1(r’(v)))\cdot \deg'(v) \vert = \vert (1+H(r(v^*)))\cdot \deg'(v^*) +(1+ H(r(u^*)))\cdot \deg'(u^*)- (1+H(r’(v^*)))\cdot \deg'(v^*)-(1+ H(r’(u^*)))\deg'(u^*) \vert \leq 2 \cdot 6M_{\rho,n} \left(3+ \beta+\frac{1}{\beta}\right) = 12 M_{\rho,n} \left(3+ \beta+\frac{1}{\beta}\right).$ Note that we introduce $\deg'(v)$, to ensure that the sensitivity of $ f_{S_1,\tilde{d}}$ remains small.
\end{proof}

 By composition, we have that the main algorithm is $\eps$-DP. 
\end{proof}
\section{Preliminaries}
We state the following tail bound for a random variable drawn from the Laplace Distribution. 

 \begin{fact}\label{fact:lap}
 If $Y \sim Lap(b)$, then 
 \[ \Pr[ \vert Y \vert \geq \ell \cdot b] = \exp(-\ell) \; .\]
 \end{fact} 
 
Next, we state a well-known fact which implies that the concentration results for sampling with replacement obtained using Chernoff bounds type methods (bounding moment generating function + Markov inequality) can be transferred to the case of sampling without replacement. 

\begin{fact}[\cite{bardenet2015concentration, hoeffding1994probability}]\label{fact:conc}
Let $\cX=(x_1, \ldots,x_N)$ be a finite population of $N$ points and $X_1, \ldots, X_n$ be a random sample drawn without replacement from $\cX$, and $Y_1, \ldots, Y_n$ be a random sample drawn with replacement from $\cX$. If $f: \bbR \to \bbR$ is continuous and convex, then 
\[ \E\left[ f \left( \sum^n_{i=1} X_i\right)\right] \leq \E\left[f \left( \sum^n_{i=1} Y_i\right)\right] \;. \] 
\end{fact}

% For e.g., the standard version of Hoeffding's inequality holds for random variables sampled without replacement. 

% \begin{fact}[\cite{bardenet2015concentration, hoeffding1994probability}]\label{fact:conc-hoeff}
% Let $\cX=(x_1, \ldots,x_N)$ be a finite population of $N$ points and $X_1, \ldots, X_n$ be a random sample drawn without replacement from $\cX$. Let $a=\min_{1 \leq i \leq N}x_i$, and $b=\max_{1 \leq i \leq N} x_i$. Then for all $\nu>0$, 
% \[ \Pr\left[ \frac{1}{n} \sum^n_{i=1}X_i - \mu \geq \nu \right] \leq \exp \left(- \frac{2n \nu^2}{(b-a)^2}\right)\;,\]
% where $\mu=\frac{1}{N} \sum^N_{i=1}x_i$ is the mean of $\cX$. 
% \end{fact}

\section{Accuracy Analysis of Theorem~\ref{thm:informal-avg-deg}} \label{sec:accuracy-avgdeg}
\subsection{Proof Sketch of Theorem~\ref{thm:informal-avg-deg}}\label{sec:avg-deg-sketch}
In this section, we give a sketch of the accuracy analysis. The more formal proofs can be found in Section~\ref{sec:formal-accuracy-avgdeg}.
\begin{theorem}\label{thm:main-accuracy-avgdeg}
For every $\rho < 1/4$, $\beta \leq \rho/8$, and $\eps^{-1} =  o(\log^{1/4}(n))$, for sufficiently large $n$, the main algorithm~(see Algorithm~\ref{alg:1+e}) outputs a value $\tilde{d}$ such that with probability at least $1-o(1)$, it holds that 
\[\left(1 -\rho \right)\cdot \bar{d} \leq \tilde{d} \leq \left(1 + \rho\right) \cdot \bar{d} \]
\end{theorem}
\begin{proof}
 The main proof strategy conditions on $S_1$ being sufficiently large or not. First, consider Case 1 when $\vert S_1 \vert < 1.2T \cdot \sqrt{\vert S \vert}\cdot \vert S \vert $ where $T$ is a size threshold parameter. We first show that for $i \in I$ the noisy buckets $\vert \tilde{B}_i \vert/n$ are approximated well by $\vert \tilde{S}_i \vert / \vert S \vert $ (see Part~\ref{it:big-approx-bi-si} of Lemma~\ref{lem:big-approx}). Next we show that the number of vertices in buckets that are significantly smaller than the size threshold are of size $O(\sqrt{n})$ (for buckets $U':=\{v \in \tilde{B}_i: (i \not \in I) \wedge (i>\log_{1+\beta} \left(\frac{6M_{\rho,n}}{\beta}\right)+2)\}$, see Part~\ref{it:small-small-u'} of Lemma~\ref{lem:small-small}) and of size $\tilde{O}(n^{3/4})$ (for bucket $B_1:=\cup_{i<\log_{1+\beta}\left(\frac{6M_{\rho,n}}{\beta}\right) +2} \tilde{B}_i$, see Part~\ref{it:small-small-b1} of Lemma~\ref{lem:small-small}). This leads to Corollary~\ref{corol:size-of-u} which bounds the number of edges between small buckets as roughly $\tilde{O}(\rho n + n^{3/4})$. 

One of our main contributions is showing that the actual fraction of edges between sufficiently large buckets and small buckets, denoted by ${\alpha}_i$, is approximated well by our noisy estimator $\tilde{\alpha_i}$ (see Lemma~\ref{lem:alpha-approx}, which implies the following corollary).

\begin{corollary}\label{corol:alpha-frac}
Assuming that $\eps^{-1} =  o(\log^{1/4}(n))$, for every $i \in I$, for sufficiently large $n$, we have that with probability at least $1-o(1)$,
\begin{enumerate}
    \item $ \vert \tilde{\alpha}_i - \alpha_i \vert \leq \frac{\rho}{4}\alpha_i $ if $\alpha_i \geq \rho/8$.
    \item $\tilde{\alpha}_i \leq \rho/4$, if $\alpha_i \leq \rho/8$. 
\end{enumerate}
\end{corollary}

Finally, we need to show that for sufficiently large noisy buckets, the actual degrees of the vertices (sans noise) only shifts to an adjacent noisy bucket (see~Lemma~\ref{lem:big-noisy-bucket-same}). This helps us bound the number of edges whose one endpoint resides in a sufficiently large noisy bucket. We have shown that with high probability, all approximations of edges between the different types of buckets is good, which leads to the main Lemma~\ref{lem:final-avg-deg-case1} for Case 1. 

Now consider Case 2 when $\vert S_1 \vert > 1.2T \cdot \sqrt{\vert S \vert}\cdot \vert S \vert $. We show that the bucket $\vert B_1 \vert/n$ is now approximated well by $\vert S_1 \vert / \vert S \vert$ (see Part~\ref{it:big-approx-b1-s1} of Lemma~\ref{lem:big-approx}). We introduce a different estimator for counting edges between $B_1$ and small buckets given by $Z+\sum_{v \in S_1} (1+X(v)) \cdot \deg'(v) $, where $Z\sim \Lap\left(36M_{\rho,n} \left(3+ \beta+\frac{1}{\beta}\right)\right)$ and $\deg'(v)= \min \{\deg(v),6M_{\rho,n} \left(3+ \beta+\frac{1}{\beta}\right)\}$. First, we show that for every $v\in S_1$, with high probability $\deg'(v)=\deg(v)$ (see Lemma~\ref{lem:deg'=degv}). Our main contribution in this case is showing that our estimator (sans noise) approximates the fraction of the sum of the edges between $B_1$ and all vertices in the graph (denoted by $E_1$), and the edges between $B_1$ and vertices in small buckets in the graph (denoted by $E'_1$) well (see lemma below). 
\begin{lemma}\label{lem:case2-output-hoeffding}
Let $\bar{d}_1$ be the average degree of bucket $B_1$. If $\vert B_1 \vert >1.5 T \cdot \sqrt{\vert S \vert } \cdot n$,  
\begin{enumerate}
\item \label{it:case2-output-hoeffding1} If $\bar{d}_1 \geq 1$, then with probability at least $1-o(1)$,
$$\left( 1-\frac{\rho}{4}\right) \cdot \frac{\vert E_1 \vert + \vert E'_1 \vert }{n} < \frac{1}{\vert S \vert} \sum_{v \in S_1} (1+X(v)) \cdot \deg(v) <  \left( 1+\frac{\rho}{4} \right) \cdot \frac{\vert E_1 \vert + \vert E'_1 \vert }{n} $$
\item \label{it:case2-output-hoeffding2} If $\bar{d}_1 < 1$, and $\bar{d}\geq 1$,  then with probability at least $1-o(1)$,
$$ \frac{\vert E_1 \vert + \vert E'_1 \vert }{n}-\rho/4 < \frac{1}{\vert S \vert} \sum_{v \in S_1} (1+X(v)) \cdot \deg(v) <  \frac{\vert E_1 \vert + \vert E'_1 \vert }{n} +\rho/4$$
\end{enumerate}
\end{lemma}
To complete this part of the proof, we show that the noise added to the estimator (denoted by $Z$) is small (see Claim~\ref{clm:Z-lap-noise-small}) and therefore, the noisy estimator also approximates the quantity $(\vert E_1 \vert +\vert E'_1 \vert)/n$ well. 

The rest of the analysis is similar to Case 1 and we invoke the same lemmas to show that with high probability, the approximations of edges between the rest of the sufficiently large buckets, and between the small buckets, as well as between the sufficiently large buckets and small buckets is good, thus giving us the main Lemma~\ref{lem:final-avg-deg-case2} for Case 2. 

Combining these two main lemmas proves our main theorem statement. 
\end{proof}

\begin{remark}
A simpler algorithm for estimating the average degree was given by Seshadri~\cite{seshadhri2015simpler}. The main intuition behind his algorithm was that out of $m$ edges of a graph, there are not ``too many'' edges that contribute a high degree. Thus the algorithm samples vertices and a random neighbor of each sampled vertex, but it only counts edges (scaled by a factor of 2 times the degree of the sampled vertex) for which the degree of the random neighbor is higher than that of the degree of the sampled vertex. 

The Coupled Global Sensitivity of the final estimate returned by this algorithm is high (proportional to the degree of the sampled vertex and its random neighbor); thus adding Laplace noise directly to the estimate would result in a very inaccurate algorithm. It is unclear how to mitigate this issue and make this algorithm differentially-private with a reasonable accuracy guarantee.
\end{remark}

%\iffalse
\subsection{Formal proofs of Theorem~\ref{thm:informal-avg-deg}. }\label{sec:formal-accuracy-avgdeg}
In this section, we give a more formal proof of Theorem~\ref{thm:main-accuracy-avgdeg}. We define the accuracy parameter as $\rho$. For ease of calculation, we set $\vert S \vert =t \cdot \frac{\log^2 (n)}{\rho^2} \cdot \sqrt{\frac{n}{\rho}} \cdot \left( 1+\frac{1}{\eps} \right)$, $T:=\frac{1}{2}\sqrt{\frac{\rho}{n}}\cdot \frac{\eps}{(1+\eps)} \cdot \frac{1}{t}$, and $M_{\rho,n} := \frac{1}{3} \cdot \sqrt{\frac{\rho}{n \sqrt{\log (n)}}} \cdot \frac{\vert S \vert }{t}$. % where $C>1$ and $C'>0$ are constants. 
 
 We define a \emph{noisy bucket} $\tilde{B}_i$ as $\{v \in G\ :\ \tilde{d}(v) \in ((1+\beta)^{i-1},(1+\beta)^i] \}$ where $\tilde{d}(v) := \deg(v) + Y_v$ where $Y_v \sim \Lap(6/\eps)$.
 We also define ${B}_1:=\cup_{i<\log_{1+\beta}\left(\frac{6M_{\rho,n}}{\beta}\right) +2} \tilde{B}_i$. 
 
 For $i >  \log_{1+\beta} \left(\frac{6M_{\rho,n}}{\beta}\right)+2$, we define a noisy bucket $\tilde{B}_i$ as \emph{big} if $\vert \tilde{B}_i \vert \geq 1.2T \cdot n = \frac{3}{5t} \sqrt{\rho n} \cdot {\eps} $. We say $B_1$ is big if $\vert B_1 \vert > 1.5T \cdot \sqrt{\vert S \vert}\cdot n$.

We first use Fact\ref{fact:lap} to bound the Laplace noise in terms of $M_{\rho,n}$ and Laplace noise parameter $p$ in Lemma~\ref{lem:lap2} and Corollary~\ref{corol:yi-small}.

\begin{lemma}\label{lem:lap2}
For $Y \sim Lap(p/\eps)$,  with probability at least $1- \exp\left( -\frac{\log^{7/4} (n) }{3\rho^2} \cdot (1+\eps) \right)$, we have that $\vert Y \vert < p \cdot M_{\rho,n}$ where $M_{\rho,n} := \frac{1}{3} \cdot \sqrt{\frac{\rho}{n \sqrt{\log (n)}}} \cdot \frac{\vert S \vert }{t}$.
\end{lemma}
%\tanote{change expression in lemma statement}
 \begin{proof}
\begin{align*}
    &\Pr\left[\vert Y \vert \geq p \cdot  \frac{1}{3} \cdot \sqrt{\frac{\rho}{n \sqrt{\log (n)}}} \cdot \frac{\vert S \vert }{t}\right] \\
    &= \Pr\left[\vert Y \vert \geq  p \cdot \frac{1}{3} \cdot \sqrt{\frac{\rho}{n \sqrt{\log (n)}}}\cdot t \cdot\frac{\log^2 (n)}{\rho^2} \cdot \sqrt{\frac{n}{\rho}} \cdot \left(1+\frac{1}{\eps} \right)\cdot \frac{1}{t}\right]&\\
    &= \Pr\left[\vert Y \vert \geq \left( \frac{\log^{7/4} (n) }{3\rho^2} \cdot \left(1+\frac{1}{\eps}\right)\eps\right) \cdot \frac{p}{\eps} \right]& \\
    &= \exp\left( -\frac{\log^{7/4} (n) }{3\rho^2} \cdot \left(1+{\eps}\right) \right)& \text{Using Fact \ref{fact:lap}} &\\
\end{align*} 
 \end{proof}

\begin{corollary}\label{corol:yi-small}
For all $i=1, \ldots, \lceil \log_{(1+\beta)} n\rceil$, if $Y_i \sim Lap(p/\eps)$, then $\vert Y_i \vert < p \cdot M_{\rho,n}$ with probability at least $1-o(1)$.   
\end{corollary}
\begin{proof}
Using Lemma \ref{lem:lap2} and a union bound we have that, 
\begin{align*}
    \Pr\left[\exists i\ :\ \vert Y_i \vert \geq p\cdot \frac{1}{3} \cdot \sqrt{\frac{\rho}{n \sqrt{\log (n)}}} \cdot \frac{\vert S \vert }{t}\right] 
    &\leq \lceil \log_{(1+\beta)} n\rceil \cdot  \exp\left( -\frac{\log^{7/4} (n) }{3\rho^2} \cdot \left(1+{\eps}\right) \right)\\
    &= \exp\left(\ln \lceil\log_{(1+\beta)}(n)\rceil - \frac{\log^{7/4} (n) }{3\rho^2} \cdot \left(1+{\eps}\right)\right)
\end{align*}
For constant $\rho$ and $\beta$, the above expression is $o(1)$.
Therefore with probability $1-o(1)$, the claim follows.  
\end{proof}

\paragraph{CASE 1: $\vert S_1 \vert < 1.2T \cdot \sqrt{\vert S \vert}\cdot \vert S \vert $.}

We follow the same strategy as outlined in the proof sketch in Section~\ref{sec:avg-deg-sketch}. 

First, we show that for sufficiently large buckets, $\vert \tilde{S}_i \vert/ \vert S \vert$ (respectively $\vert S_1 \vert/\vert S \vert$) approximates $\vert \tilde{B}_i \vert/n$ (respectively $\vert B_1 \vert/n$) well. We only need Part~\ref{it:big-approx-bi-si} in this case, and use Part~\ref{it:big-approx-b1-s1} in the analysis of Case 2. 

\begin{lemma} \label{lem:big-approx}
%Let $\rho<1/2$, 
\begin{enumerate}
    \item \label{it:big-approx-bi-si} For all $i>\log_{1+\beta} \left(\frac{6M_{\rho,n}}{\beta}\right)+2$ such that $\vert \tilde{B}_i \vert \geq 1.2T \cdot n$, then with probability at least $1-o(1)$, 
 \[ \left(1-\frac{\rho}{4}\right) \cdot \frac{\vert \tilde{B}_i \vert }{n} \leq \frac{\vert \tilde{S}_i \vert }{\vert S \vert} \leq \left(1+\frac{\rho}{4}\right)\cdot \frac{\vert \tilde{B}_i \vert }{n}\]
Otherwise, if $\vert \tilde{B}_i \vert < T \cdot n$, then with probability at least $1-o(1)$, we have that $\vert \tilde{S}_i \vert < 1.1T\cdot \vert S \vert$. 
\item \label{it:big-approx-b1-s1} If $\vert B_1 \vert >1.5 T \cdot \sqrt{ \vert S \vert } \cdot n$, then with probability at least $1-o(1)$, we have 
\[(1-\rho/4) \cdot \frac{\vert B_1 \vert}{n}  < \frac{\vert S_1 \vert}{\vert S \vert} < (1+\rho/4) \cdot \frac{\vert B_1 \vert}{n} \;,\]
Otherwise, if $\vert B_1 \vert < T \cdot \sqrt{ \vert S \vert } \cdot n$, then with probability at least $1-o(1)$, we have that $\vert S_1 \vert < 1.1T\cdot \sqrt{\vert S \vert} \cdot \vert S \vert $. 
\end{enumerate}
\end{lemma}

\begin{proof}
The proofs for both parts are very similar, we include both proofs here for completeness. 
\begin{enumerate}
    \item Let $X_{v} = 1$ if the sampled vertex $v$ is in noisy big bucket $\tilde{B}_i$, and 0 otherwise. Clearly $\vert \tilde{S}_i \vert = \sum_{v \in S} X_{v}$. 
    
{\bf Case 1: $\vert \tilde{B}_i \vert \geq 1.2T \cdot n$.} Thus using Chernoff bounds, and recalling that $\E[\tilde{S_i}] = \vert S \vert \frac{\vert \tilde{B}_i \vert }{n} \geq \frac{3\log^2(n)}{5\rho^2}$,

 \begin{align*}
    \Pr[\vert \vert \tilde{S}_i \vert - \E[\vert \tilde{S}_i \vert]\vert \geq \frac{\rho}{4} \E[\vert \tilde{S}_i \vert]  ]
    &\leq 2 \exp\left( - \frac{\rho^2}{16} \cdot \frac{1}{3} \cdot \E[\vert \tilde{S}_i \vert] \right)\\
    &\leq 2 \exp\left( - \frac{\rho^2}{16} \cdot \frac{1}{3} \cdot \frac{3\log^2(n)}{5 \rho^2} \right)
 %   &=  \exp\left( \ln 2 - \frac{1}{80}\cdot {\log^2(n)}\right)\\
\end{align*}
{\bf Case 2: $\vert \tilde{B}_i \vert < T \cdot n$.} 
Observe that, 
\begin{align*}
    Var[\vert \tilde{S}_i \vert] &= \sum_{v \in S} Var[X_v] + \sum_{v \in S}\sum_{\substack{v' \in S \\v'\neq v}} Cov(X_v,X_{v'})\\
    &\leq \vert S \vert \left( \frac{\vert \tilde{B}_i \vert}{n} -\frac{\vert \tilde{B}_i \vert^2}{n^2}  \right)+ \frac{\vert S \vert (\vert S \vert -1 ) \vert \tilde{B}_i\vert^2}{n^2(n-1)}\\
    &< \vert S \vert T+ \frac{\vert S \vert (\vert S\vert -1)T^2}{n-1}
\end{align*}
where we used the fact that,
\begin{align*}
    Cov(X_v,X_{v'}) = \E[X_v \cdot X_v'] - \E[X_v]\cdot \E[X_{v'}] = \frac{\vert \tilde{B}_i \vert }{n}\cdot \frac{\vert \tilde{B}_i \vert-1 }{n-1} - \frac{\vert \tilde{B}_i \vert^2 }{n^2} \leq \vert \tilde{B}_i \vert^2 \cdot \left(\frac{1}{n(n-1)} - \frac{1}{n^2} \right) \;.
\end{align*}

By Chebyshev's inequality, we have, 
\begin{align*}
    &\Pr[\vert \vert \tilde{S}_i \vert - \E \vert \tilde{S}_i \vert \vert \geq 0.1 T \vert S \vert ] \\
    &\leq \frac{Var[\vert \tilde{S}_i \vert]}{(0.1T \vert S \vert)^2 }\\
    &\leq \frac{1}{(0.1)^2 \vert S \vert T} + \left(1-\frac{1}{\vert S \vert}\right)\cdot \frac{100}{n-1}\\
    &= \frac{200\rho^2}{\log^2(n)} + \left(1-\frac{\rho^2}{t \log^2(n)}\cdot \sqrt{\frac{\rho}{n}} \cdot \frac{\eps}{1+\eps}\right)\frac{100}{n-1} = o(1)
\end{align*}
Thus with probability at least $1-o(1)$, $\vert \tilde{S}_i \vert < \E \vert \tilde{S}_i \vert \vert+ 0.1 T \vert S \vert  < T\cdot  \vert S \vert (1+0.1) < 1.1T\cdot \vert S \vert $.

\item Let $X_{v} = 1$ if the sampled vertex $v$ is in bucket $\tilde{B}_1$, and 0 otherwise. Clearly $\vert{S}_1 \vert = \sum_{v \in S} X_{v}$, and $\E[\vert S_1 \vert ]=\frac{\vert S \vert \cdot \vert B_1\vert }{n}  $. 

{\bf Case 1: $\vert B_1 \vert > 1.5T \cdot \sqrt{\vert S \vert } \cdot n$.} By Chernoff bounds we have, 
\begin{align*}
    &\Pr[\vert \vert S_1 \vert - \E[\vert S_1 \vert]\vert \geq \frac{\rho}{4} \E[\vert S_1 \vert]  ] \\
    &\leq 2 \exp\left( - \frac{\rho^2}{16} \cdot \frac{1}{3} \cdot \E[\vert S_1 \vert] \right)\\
    &\leq 2 \exp\left( - \frac{\rho^2}{16} \cdot \frac{1}{3} \cdot \vert S \vert \cdot \frac{\vert B_1 \vert }{n}\right)\\
    &\leq 2 \exp\left( - \frac{\rho^2}{16} \cdot \frac{1}{3} \cdot \vert S \vert \cdot \frac{1.5T \cdot \sqrt{ \vert S \vert } \cdot n}{n} \right)\\
    &= 2 \exp\left( - \frac{1.5\rho^2}{16} \cdot \frac{1}{3} \cdot \frac{ \log^2(n)}{2\rho^2}  \cdot \frac{t^{1/2} \cdot \log(n)}{\rho} \cdot {\frac{n^{1/4}}{\rho^{1/4}}} \cdot \sqrt{1+\frac{1}{\eps}}  \right) \\ 
    &=2\exp \left(-\frac{\sqrt{1+1/\eps}}{64 \rho^{5/4} \cdot \log^{1/2}(1+\beta)} \cdot \log^{7/2}(n) \cdot n^{1/4} \right)
\end{align*}
{\bf Case 2: $\vert B_1 \vert < T \cdot \sqrt{ \vert S \vert } \cdot n$.}
Observe that 
\begin{align*}
    Var[\vert \tilde{S}_1 \vert] &= \sum_{v \in S} Var[X_v] + \sum_{v \in S}\sum_{\substack{v' \in S \\v'\neq v}} Cov(X_v,X_{v'})\\
    &= \vert S \vert \left( \frac{\vert {B}_1 \vert}{n} -\frac{\vert {B}_1\vert^2}{n^2}  \right)+ \frac{\vert S \vert (\vert S \vert -1 ) \vert {B}_1\vert^2}{n^2(n-1)}\\
   &< \vert S \vert T \sqrt{ \vert S \vert }+ \frac{\vert S \vert^2 (\vert S\vert -1)T^2}{n-1}
\end{align*}
 
By Chebyshev's inequality, we have, 
\begin{align*}
    &\Pr[\vert \vert \tilde{S}_i \vert - \E \vert \tilde{S}_i \vert \vert \geq 0.1 T \sqrt{ \vert S \vert } \vert S \vert ] \\
    &\leq \frac{Var[\vert \tilde{S}_i \vert]}{(0.1T  \sqrt{ \vert S \vert }\vert S \vert)^2 }\\
    &\leq \frac{1}{(0.1)^2 \vert S \vert  \sqrt{ \vert S \vert }T} + \left(1-\frac{1}{\vert S \vert}\right)\cdot \frac{1}{n-1}\\
    &= \frac{200\rho^{9/2} \sqrt{\log(1+\beta)}}{\sqrt{1+1/\eps}\sqrt{n}\log^{7/2}(n)} + \left(1-\frac{\rho^2 \log (1+\beta)}{\log^3(n)}\cdot \sqrt{\frac{\rho}{n}} \cdot \frac{\eps}{1+\eps}\right)\frac{1}{n-1} = o(1)
\end{align*}

Thus with probability at least $1-o(1)$, $\vert S_1 \vert <\E \vert S_1 \vert + 0.1 T\cdot \sqrt{\vert S \vert}\cdot \vert S \vert< (T\sqrt{\vert S \vert} \vert S \vert (1+0.1) =1.1T\sqrt{\vert S \vert }\vert S \vert$.
\end{enumerate}
\end{proof}

We reuse the following notation from~\cite{goldreich2004estimating}, 
 $$ E(V_1,V_2) := \{ (v_1,v_2)\ :\ v_1 \in V_1\ \&\ v_2 \in V_2\ \&\ \{v_1,v_2\} \in E(G)\}$$
 In other words, $E(V_1,V_2)$ denotes the set of all ordered pairs of adjacent vertices, with the first vertex in $V_1$ and the second vertex in $V_2$. 

Define $E_i$ as the set of ordered pairs of adjacent vertices such that the first vertex is in $\tilde{B}_i$, i.e., $E_i := E(\tilde{B}_i,V)$. 
Let $U:= \{v \in \tilde{B}_i: i \not \in I\}$, i.e., $U$ is the set of vertices that reside in noisy buckets deemed ``small'' by the sample. 

Since $\vert S_1 \vert < 1.2T \cdot \sqrt{\vert S \vert}\cdot \vert S \vert $ , we have that $U = U' \cup B_1$, where $U':=\{v \in \tilde{B}_i: (i \not \in I) \wedge (i>\log_{1+\beta} \left(\frac{6M_{\rho,n}}{\beta}\right)+2)\}$. 
We also define the set of edges between a noisy bucket and $U$ as $E'_i$, i.e., $E'_i := E(\tilde{B}_i,U) \subseteq E_i$. Also, $E_1 := E({B}_1,V)$, and $E'_1 := E({B}_1,U)$. 
Then 
$$ \sum_{i \in I} \vert E'_i \vert = E (V \setminus U ,  U ),\ \sum_{i \in I} \vert E_i \setminus E'_i \vert = 2 \vert E(V \setminus U, V \setminus U) \vert $$

The next Lemma bounds the number of vertices in small buckets $U'$ and $B_1$, and the subsequent Corollary bounds the total number of edges in all the small buckets, denoted by $U$.
\begin{lemma}\label{lem:small-small}
Define the sets $U':= \{v \in \tilde{B}_i: (i \not \in I) \wedge (i>\log_{1+\beta} \left(\frac{6M_{\rho,n}}{\beta}\right)+2)\}$, and ${B}_1:=\cup_{i<\log_{1+\beta}\left(\frac{6M_{\rho,n}}{\beta}\right) +2} \tilde{B}_i$. Then with probability at least $1-o(1)$, 
\begin{enumerate}
\item  \label{it:small-small-u'}
$ \vert U' \vert \leq \frac{3}{4}\cdot \sqrt{\rho n} \;.$
    \item \label{it:small-small-b1}
$ \vert B_1 \vert < \frac{3 \cdot \sqrt{\log(1+\beta)}}{4 \rho^{3/4}  \sqrt{1+1/\eps}} \cdot n^{3/4} \cdot\log^{1/2}(n)  \;. $
\end{enumerate}
\end{lemma}
\begin{proof}
\begin{enumerate}
\item  Using Lemma~\ref{lem:big-approx}, Item~\ref{it:big-approx-bi-si}, we know that if $\vert \tilde{B}_i \vert \geq 1.2T \cdot n$, then with probability at least $1-o(1)$, ${\vert \tilde{S}_i \vert } \geq \left(1-\frac{\rho}{4}\right) \cdot \frac{\vert \tilde{B}_i \vert}{n} \cdot {\vert S \vert}\geq 1.2T \cdot {\vert S \vert} $. Therefore, we have that with probability at least $1-o(1)$, 
 $$ \vert U' \vert \leq \vert \{ v \in \tilde{B}_i\ :\ \vert \tilde{B}_i \vert < 1.2 T \cdot n \} \vert \leq t \cdot1.2 T \cdot n =\frac{3}{5}\cdot \sqrt{\rho n} \cdot \frac{\eps}{1+\eps} < \frac{3}{5} \sqrt{\rho n}$$
\item Using Lemma~\ref{lem:big-approx}, Item~\ref{it:big-approx-b1-s1}, we know that if $\vert B_1 \vert \geq 1.5T \cdot \sqrt{\vert S \vert } \cdot n$, then with probability at least $1-o(1)$, we have $\vert S_1 \vert \geq  \left(1-\frac{\rho}{4}\right) \frac{\vert B_1 \vert}{n} \cdot \vert S \vert >1.5 T \cdot \sqrt{\vert S \vert } \cdot \vert S \vert$. Therefore with probability at least $1-o(1)$, 
\begin{align*}
\vert B_1 \vert &\leq 1.5T \cdot \sqrt{\vert S \vert} \cdot n < \frac{3\cdot \sqrt{\log(1+\beta)}}{4\rho^{3/4}  \sqrt{1+1/\eps}} \cdot n^{3/4} \cdot\log^{1/2}(n) 
\end{align*}
\end{enumerate}
\end{proof}

\begin{corollary}\label{corol:size-of-u}
Let $U:= U' \cup B_1$, then with probability at least $1-o(1)$,
$$ \vert E( U,U) \vert < \frac{9}{25}\cdot {\rho n} + \frac{3\rho^{-11/4}}{2}\cdot {(2+1/\beta+\beta)\sqrt{\log(1+\beta)}} \cdot \sqrt{1+\frac{1}{\eps}}\cdot n^{3/4} \log^{9/4}(n)$$
%$$ \vert E( U,U) \vert < \frac{3}{4}\cdot {\rho n} \cdot {\eps^2}$$
\end{corollary}
\begin{proof}
Using Lemma~\ref{lem:small-small}, we know that with probability $1-o(1)$,
\begin{align*}
    &\vert E( U,U) \vert\\ &\leq \vert U'\vert^2 + \vert B_1 \vert \cdot (\text{max deg of a vertex in }B_1) \\
    &\leq  \left(\frac{3}{5}\cdot \sqrt{\rho n} \right)^2 +  \left(\frac{3\cdot \sqrt{\log(1+\beta)}}{4 \rho^{3/4}  \sqrt{1+1/\eps}}  \cdot n^{3/4} \cdot\log^{1/2}(n) \right) \cdot \frac{6M_{\rho,n}}{\beta} (1+\beta)^2\\
    &\leq \frac{9}{25}\cdot {\rho n} + \frac{3\rho^{-11/4}}{2}\cdot {(2+1/\beta+\beta)\sqrt{\log(1+\beta)}} \cdot \sqrt{1+\frac{1}{\eps}}\cdot n^{3/4} \log^{9/4}(n) \\
  %  &\leq \frac{9}{16}\cdot {\rho n} \cdot {\eps^2} +  \frac{3}{16}\cdot {\rho n} \cdot {\eps^2}
\end{align*}
\end{proof}

Lemma~\ref{lem:alpha-approx} shows that our noisy estimator for approximating the fraction of edges between sufficiently large buckets and small buckets denoted by $\tilde{\alpha}_i$ is good. This is one of our main contributions. We first introduce a claim about the Laplace noise term used in our estimator which we use to bound the noise term in the proof of Lemma~\ref{lem:alpha-approx}. 

\begin{claim}\label{clm:zij-si}
Let $X_i \sim \Lap(6/\eps)$, for every $i \in I$, with probability at least $1-o(1)$, 
\[ \left\vert \frac{X_{i}}{\vert \tilde{S}_i \vert } \right\vert < \frac{10}{3}\left(1+\frac{1}{\eps} \right)\cdot \log^{-\frac{1}{4}}(n) \]
%where $i \in I$ and $X_i \sim \Lap(4/\eps)$.
\end{claim}

\begin{proof}
 Using Corollary \ref{corol:yi-small}, we have that with probability $1- \exp\left(\ln \lceil\log_{(1+\beta)}(n)\rceil - \frac{\log^{7/4} (n) }{3\rho^2} \cdot \left(1+{\eps}\right)\right)$, for every $i \in I$ such that $X_{i} \sim \Lap(6/\eps)$,
\begin{align*}
    \left\vert \frac{X_{i}}{\vert \tilde{S}_i \vert} \right\vert &< \frac{6M_{\rho,n}}{\vert \tilde{S}_i\vert} \leq \frac{10}{3}\left(1+\frac{1}{\eps} \right)\cdot \log^{-\frac{1}{4}}(n)
\end{align*}

By union bound, 
\begin{align*}
    \Pr\left[ \exists\ i:\ \frac{X_{i}}{\vert \tilde{S}_i \vert} >\frac{10}{3}\left(1+\frac{1}{\eps} \right)\cdot \log^{-\frac{1}{4}}(n)\right] 
    &\leq \lceil \log_{1+\beta}(n) \rceil  \exp\left(\ln \lceil\log_{(1+\beta)}(n)\rceil - \frac{\log^{7/4} (n) }{3\rho^2} \cdot \left(1+{\eps}\right)\right)\\
    &= \exp\left(2\ln \lceil \log_{(1+\beta)}(n) \rceil- \frac{\log^{7/4} (n) }{3\rho^2} \cdot  \left(1+{\eps}\right)\right)
\end{align*}
\end{proof}

 \begin{lemma}\label{lem:alpha-approx}
With probability at least $1-o(1)$, for every $i \in I$, for $\alpha_i:=\frac{\vert E'_i \vert}{\vert E_i \vert}$, and 
$\tilde{\alpha}_i := \frac{W_i}{\tilde{S}_i}$ we have

\begin{enumerate}
\item \label{it:alpha-approx-pt1} % $1-\exp \left(\log 2 - \log^{11/4}(n) (3 \log^{1/4}(n)-2) \right)$,
$ \vert \tilde{\alpha}_i - \alpha_i \vert \leq \frac{\rho}{4}\alpha_i - \frac{10}{3}\left(1+\frac{1}{\eps} \right) \log^{-1/4}(n)$, if $\alpha_i \geq \rho/8$, or
%$ \left(1-\frac{\rho}{4}\right) \cdot \alpha_i \leq \tilde{\alpha}_i \leq \left(1+\frac{\rho}{4}\right) \cdot \alpha_i $, if $\alpha_i \geq \rho/8$, or
\item \label{it:alpha-approx-pt2} $\vert \tilde{\alpha}_i - \alpha_i \vert > \rho/16-\frac{10}{3}\left(1+\frac{1}{\eps} \right)\log^{-1/4}(n)$, %$\tilde{\alpha}_i < \rho/4$, 
if $\alpha_i < \rho/8$.
\end{enumerate}

\end{lemma}
\begin{proof} We first prove the claim in Part~\ref{it:alpha-approx-pt1}, and then the claim in Part~\ref{it:alpha-approx-pt2}.
\begin{enumerate}
    \item 
For a fixed $i\in \{1,\ldots, \log_{1+\beta}(n)\}$, we define BAD$_i$ to be the event that all assumptions hold, i.e., $i \in I$ and $\alpha_i \geq \rho/8$; but $ \vert \tilde{\alpha}_i - \alpha_i \vert > \frac{\rho}{4}\alpha_i -\frac{10}{3}\left(1+\frac{1}{\eps} \right)\log^{-1/4}(n)$. Then we can define BAD to be the event that there exists an $i$ such that BAD$_i$ occurs. By union bound,
\begin{align}
\Pr[\text{BAD}]\leq \lceil \log_{1+\beta}(n) \rceil \cdot \max_i\ \Pr[\text{BAD}_i] 
\end{align}
Now we just need to give an upper bound for the probability of BAD$_i$ occurring. Observe that, 
\begin{align}
&\Pr[\text{BAD}_i] \nonumber \\ &\leq \Pr\left[\left\vert\frac{Z_i}{\vert \tilde{S}_i \vert } \right\vert \geq \frac{10}{3}\left(1+\frac{1}{\eps} \right)\cdot \log^{-\frac{1}{4}}(n)\right]
+ \Pr\left[\text{BAD}_i\ |\ \left\vert\frac{Z_i}{\vert \tilde{S}_i \vert } \right\vert \leq \frac{10}{3}\left(1+\frac{1}{\eps} \right)\cdot \log^{-\frac{1}{4}}(n)\right] \label{eq:bad-sum}
\end{align}
From Claim \ref{clm:zij-si}, we already have an upper bound for the first term in Equation \ref{eq:bad-sum}. For the rest of this proof, we will focus on upper bounding the second term.

Recall from Algorithm \textbf{NoisyBigSmallEdgeCount} (see Algorithm~\ref{alg:bigsmall}), for every $i \in I$, we defined $X(v)$ as a r.v. for every $v \in \tilde{S}_i$, defined as 
\[
X(v) =\begin{cases}
1 & \text{if random neighbor of }v\text{ belongs to a small noisy bucket}\\
0 & \text{otherwise}
\end{cases}
\]

Also recall that $ W_i := \sum_{v \in \tilde{S}_i} X(v) + Z_i$  where $Z_i \sim \Lap(6/\eps)$ and $\tilde{\alpha}_i := \frac{W_i}{\vert \tilde{S}_i \vert }$. We define $\alpha_i^* := \tilde{\alpha}_i - \frac{Z_i}{\vert \tilde{S}_i \vert } = \frac{\sum_{v \in \tilde{S}_i} X(v)}{\vert \tilde{S}_i \vert}$.
\begin{claim}\label{clm:alpha-star-chernoff}
Let $\alpha_i \geq \rho/8$, and $\alpha_i^* := \frac{\sum_{v \in \tilde{S}_i} X(v)}{\vert \tilde{S}_i \vert}$. With probability at least $1-o(1)$, 
\[ \vert \alpha_i^* - \alpha_i \vert \leq \frac{\rho}{4} \cdot \alpha_i\]
\end{claim}
\begin{proof}
Observe that $\E[\alpha^*_i]= \alpha_i$, therefore, using Chernoff bounds, 
\begin{align}
\Pr \left[ \vert \alpha_i^* - \alpha_i \vert > (\rho/4)\alpha_i  \right] 
    &=\Pr \left[ \left\vert \sum_{v \in \tilde{S}_i} X(v) -\E \sum_{v \in \tilde{S}_i} X(v) \right\vert \geq  \frac{\rho}{4}  \E \sum_{v \in \tilde{S}_i} X(v) \right]  \\
    &\leq 2 \exp \left(- (\rho^2/8)\cdot \alpha^2_i \cdot \vert \tilde{S}_i \vert \right) \\
    &\leq 2 \exp\left( - \left( \frac{\rho^2}{8}\right) \cdot  \left( \frac{\rho^2}{64}\right)  \cdot \frac{3\log^2(n)}{5\rho^2}\right) 
    \label{eq:chernoff-rhs-alphai}\\
    &= \exp\left(\ln (2) -  \frac{3\rho^{2}}{2560}  \cdot {\log^2(n)}\right) 
\end{align}

Where we obtain Step~\ref{eq:chernoff-rhs-alphai} by using the assumption that $\alpha_i \geq \rho/8$, and since $i\in I$, we know $\vert \tilde{S}_i \vert \geq 1.2T \vert S \vert = \frac{3\log^2(n)}{5\rho^2}$. 
\end{proof}
Replacing $\alpha^*_i$ with $\tilde{\alpha}_i-\frac{Z_i}{\vert \tilde{S}_i \vert}$ and conditioning on $\left\vert\frac{Z_i}{\tilde{S}_i} \right\vert \leq \frac{10}{3}\left(1+\frac{1}{\eps} \right) \cdot \log^{-\frac{1}{4}}(n)$, 
\begin{align}
    \Pr \left[ \vert \tilde{\alpha}_i - \alpha_i \vert > \frac{\rho}{4}\alpha_i - \frac{10}{3}\left(1+\frac{1}{\eps} \right)\log^{-1/4}(n) \right] \leq \exp\left(\ln (2) -  \frac{3\rho^{2}}{2560}  \cdot {\log^2(n)}\right) \label{eq:alpha-zi}
\end{align}

Now we can bound $\Pr[\text{BAD}_i]$ as follows,
\begin{align*}
&\Pr[\text{BAD}_i] \nonumber &\\ 
&\leq \Pr\left[\left\vert\frac{Z_i}{\tilde{S}_i} \right\vert \geq \frac{10}{3}\left(1+\frac{1}{\eps} \right)\cdot \log^{-\frac{1}{4}}(n)\right]
+ \Pr\left[\text{BAD}_i\ |\ \left\vert\frac{Z_i}{\tilde{S}_i} \right\vert \leq \frac{10}{3}\left(1+\frac{1}{\eps} \right)\cdot \log^{-\frac{1}{4}}(n)\right] &\\
&\leq  \exp\left(2\ln \lceil \log_{(1+\beta)}(n) \rceil- \frac{\log^{7/4} (n) }{3\rho^2} \cdot  \left(1+{\eps}\right)\right) + \Pr\left[\text{BAD}_i\ |\ \left\vert\frac{Z_i}{\tilde{S}_i} \right\vert \leq \frac{10}{3}\left(1+\frac{1}{\eps} \right)\cdot \log^{-\frac{1}{4}}(n)\right] &\text{From Claim~\ref{clm:zij-si}}\\
&\leq \exp\left(2\log \lceil \log_{(1+\beta)}(n) \rceil- \frac{\log^{7/4} (n) }{3\rho^2} \left(1+{\eps}\right)\right) + \exp\left(\ln (2) -  \frac{3\rho^{2}}{2560}  \cdot {\log^2(n)}\right) &\text{Using~Eq.~\ref{eq:alpha-zi}}
\end{align*}

Finally, 
\begin{align}
&\Pr[\text{BAD}] \nonumber\\ 
&\leq \lceil \log_{1+\beta}(n) \rceil \cdot \max_i\ \Pr[\text{BAD}_i] \\
&\leq \lceil \log_{1+\beta}(n) \rceil \cdot \left(\exp\left(2\log \lceil \log_{(1+\beta)}(n) \rceil- \frac{\log^{7/4} (n) }{3\rho^2} \left(1+{\eps}\right)\right) + \exp\left(\ln (2) -  \frac{3\rho^{2}}{2560}  \cdot {\log^2(n)}\right)\right)
\end{align}

Part~\ref{it:alpha-approx-pt1} of the theorem statement follows.
\item Next we prove Part~\ref{it:alpha-approx-pt2} of the statement. %, i.e., 

First, consider $\alpha^*_i=\tilde{\alpha}_i - \frac{Z_i}{\vert \tilde{S}_i \vert } = \frac{\sum_{v \in \tilde{S}_i} X(v)}{\vert \tilde{S}_i \vert} $, by Chernoff bounds, 
\begin{align*}
    \Pr[\vert \alpha_i^* - \alpha_i \vert > \rho/16] &= \Pr\left[\left\vert \sum_{v \in \tilde{S}_i} X(v) -\E[\sum_{v \in \tilde{S}_i} X(v)] \right\vert > (\rho/16)\cdot \vert \tilde{S}_i \vert \right]  \\
    &\leq 2 \exp \left(-\frac{2 ((\rho/16) \cdot \vert \tilde{S}_i \vert )^2}{\vert \tilde{S}_i \vert}\right)\\
    &= 2 \exp \left( -\frac{3}{640} \cdot \log^2(n)\right)
\end{align*}
where we used the fact that $\vert \tilde{S}_i \vert \geq \frac{3\log^2(n)}{5\rho^2}$. Replacing $\alpha^*_i$ with $\tilde{\alpha}_i-\frac{Z_i}{\vert \tilde{S}_i \vert}$, 
\begin{align}
    \Pr\left[\left\vert \tilde{\alpha}_i -\frac{Z_i}{\vert S_i \vert} - \alpha_i \right\vert > \rho/16\right] &\leq \exp \left(\ln 2 -\frac{3}{640} \cdot \log^2(n)\right)
\end{align}
Conditioning on $\left\vert\frac{Z_i}{\tilde{S}_i} \right\vert \leq \frac{10}{3}\left(1+\frac{1}{\eps} \right) \cdot \log^{-\frac{1}{4}}(n)$,  
\begin{align}
    \Pr \left[ \vert \tilde{\alpha}_i - \alpha_i \vert > \frac{\rho}{16} - \frac{10}{3}\left(1+\frac{1}{\eps} \right) \cdot \log^{-\frac{1}{4}}(n) \right] \leq \exp \left(\ln 2 - \frac{3}{640}  \cdot \log^2(n)\right) \label{eq:alpha-zi-case2}
\end{align}
Now, as before, 
\begin{align}
&\Pr[\vert \tilde{\alpha}_i - \alpha_i \vert > \rho/16- \frac{10}{3}\left(1+\frac{1}{\eps} \right) \cdot \log^{-\frac{1}{4}}(n)] \nonumber& \\
&\leq\Pr\left[\left\vert\frac{Z_i}{\tilde{S}_i} \right\vert \geq \frac{10}{3}\left(1+\frac{1}{\eps} \right) \cdot \log^{-\frac{1}{4}}(n)\right] \nonumber&\\
&+ \Pr\left[\vert \tilde{\alpha}_i - \alpha_i \vert > \frac{\rho}{16} - \frac{10}{3}\left(1+\frac{1}{\eps} \right) \cdot \log^{-\frac{1}{4}}(n)\ |\ \left\vert\frac{Z_i}{\tilde{S}_i} \right\vert \leq \frac{10}{3}\left(1+\frac{1}{\eps} \right) \cdot \log^{-\frac{1}{4}}(n)\right]&\\
&\leq \exp\left(2\ln \lceil \log_{(1+\beta)}(n) \rceil- \frac{\log^{7/4} (n) }{3\rho^2} (1+\eps)\right) \nonumber&\\&+ \Pr\left[\vert \tilde{\alpha}_i - \alpha_i \vert > \frac{\rho}{16} - \frac{10}{3}\left(1+\frac{1}{\eps} \right) \cdot \log^{-\frac{1}{4}}(n)\ |\ \left\vert\frac{Z_i}{\tilde{S}_i} \right\vert \leq \frac{10}{3}\left(1+\frac{1}{\eps} \right) \cdot \log^{-\frac{1}{4}}(n)\right] &\label{eq:alpha-case2-clm7}\\
&\leq \exp\left(2\ln \lceil \log_{(1+\beta)}(n) \rceil- \frac{\log^{7/4} (n) }{3\rho^2}(1+\eps) \right) +  \exp \left(\ln 2 - \frac{3}{640} \cdot \log^2(n)\right) &\label{eq:alpha-case2-zi}
\end{align}
where Eq.~\ref{eq:alpha-case2-clm7} follows from Claim~\ref{clm:zij-si}, and Eq.~\ref{eq:alpha-case2-zi} follows from substituting Eq.~\ref{eq:alpha-zi-case2}. 
Finally, by a union bound, the probability that there exists an $i$ such that $\vert \tilde{\alpha}_i - \alpha_i \vert > \rho/16$ is at most $$\lceil \log_{1+\beta}(n) \rceil \cdot  \left(\exp\left(2\log \lceil \log_{(1+\beta)}(n) \rceil- \frac{\log^{7/4} (n) }{3\rho^2} (1+\eps) \right) +  \exp \left(\ln 2 - \frac{3}{640} \cdot \log^2(n)\right)\right) \;.$$ 
\end{enumerate}

\end{proof}

Corollary~\ref{corol:alpha-frac} directly follows from Lemma~\ref{lem:alpha-approx}. 

In the following lemma, we show that the actual degrees of vertices in noisy buckets $\tilde{B}_i$ such that $i> \log_{1+\beta} \left(\frac{6M_{\rho,n}}{\beta}\right)+2$ are close to the noisy degrees.   
\begin{lemma}\label{lem:big-noisy-bucket-same}
For noisy bucket $\tilde{B}_i$ such that $i > \log_{1+\beta} \left(\frac{6M_{\rho,n}}{\beta}\right)+2$,
with probability at least $1-o(1)$, we have 
$$ (1+\beta)^{i-2} < \deg(v) \leq (1+\beta)^{i+1}\;.$$
%where $\tilde{d}(v) = \deg(v) + Y_v$ and $Y_v \sim \Lap(6/\eps)$. 
\end{lemma}
\begin{proof}
Using Corollary \ref{corol:yi-small}, with probability $1-o(1)$, we have that 
\begin{align*}
    (1+\beta)^{i-1} - 6M_{\rho,n} < \deg(v) \leq (1+\beta)^i+6 M_{\rho,n}
\end{align*}
Also, by our assumption of $i > \log_{1+\beta} \left(\frac{6M_{\rho,n}}{\beta}\right)+2$, we have that $6M_{\rho,n} < \beta(1+\beta)^{i-2}$. Therefore, 
\begin{align*}
    &(1+\beta)^{i-1} - \beta(1+\beta)^{i-2} < \deg(v) \leq (1+\beta)^i+\beta(1+\beta)^{i-2} \\
    &(1+\beta)^{i-2} < \deg(v) \leq (1+\beta)^{i+1} 
\end{align*}
\end{proof}

So far, we have shown that with high probability, the approximation of edges between the different types of buckets is good. Lemma~\ref{lem:final-avg-deg-case1} shows that the average degree of the graph is estimated well for Case 1.  

\begin{lemma}\label{lem:final-avg-deg-case1}
For every $\rho < 1/4$, $\beta \leq \rho/8$, and $\eps^{-1} =  o(\log^{1/4}(n))$, for sufficiently large $n$, and for the case when $\vert S_1 \vert < 1.2T \cdot \sqrt{\vert S \vert}\cdot \vert S \vert $, the main algorithm~(see Algorithm~\ref{alg:1+e}) outputs a value $\tilde{d}$ such that with probability at least $1-o(1)$, it holds that 
\[\left(1 -\rho \right) \bar{d} \leq \tilde{d} \leq \left(1 + \rho\right)  \bar{d} \]
\end{lemma}
\begin{proof}
Recall that $E_i$ is the set of edges consisting of ordered pairs of vertices such that the first vertex is in noisy bucket $\tilde{B}_i$. Using Lemma~\ref{lem:big-noisy-bucket-same}, for $i > \log_{1+\beta} \left(\frac{6M_{\rho,n}}{\beta}\right)+2$, we have that with probability at least $1-o(1)$,
\begin{align}
     \vert \tilde{B}_i \vert \cdot (1+\beta)^{i-2} < \vert E_i \vert < \vert \tilde{B}_i \vert \cdot (1+\beta)^{i+1} \label{eq:B-i-E-i}
\end{align}

Since the noisy buckets partition the set of edges, observe that
\begin{align}
    \bar{d} n &= 2 \vert E(V \setminus U, V \setminus U) \vert + 2 \vert E(V \setminus U, U )\vert +2 \vert E(U,U)\vert  \nonumber\\
    &\leq 2 \vert E(V\setminus U, V \setminus U) \vert + 2 \vert E(V\setminus U, U )\vert + \vert U\vert^2 \label{eq:avg-deg-edges-num}
\end{align}

Also, 
\begin{align}
    \sum_{i \in I} \vert E'_i\vert = \vert E(V \setminus U, U) \vert \label{eq:sum-Ti} \\
    \sum_{i \in I} \vert E_i \setminus E'_i \vert = 2 \vert E(V \setminus U, V \setminus U) \vert  \label{eq:sum-Ei-Ti} 
\end{align}

Thus with high probability, the following holds,

\begin{align*} 
    \tilde{d} 
    &= \frac{1}{\vert S \vert } \sum_{i \in I} \vert \tilde{S}_i \vert \cdot (1+\tilde{\alpha}_i) \cdot (1+\beta)^{i} &\\
    &\leq \frac{1}{n} \cdot \sum_{i \in I} (1+\tilde{\alpha}_i) \cdot\left(1+\frac{\rho}{4} \right) \cdot \vert \tilde{B}_i \vert \cdot (1+\beta)^i &\text{By Lemma}~\ref{lem:big-approx}\\
    &\leq \frac{(1+\rho/4)}{n} \cdot \sum_{i \in I} (1+\tilde{\alpha}_i) \cdot (1+\beta)^2 \cdot \vert E_i \vert  &\text{Using Equation}~\ref{eq:B-i-E-i}\\
    &\leq \frac{(1+\rho/4)(1+\beta)^2}{n} \cdot \bigl( \sum_{\substack{i \in I\\ \alpha_i \geq \rho/8}} (1+(1+\rho/4) \alpha_i) \cdot \vert E_i \vert +\sum_{\substack{i \in I\\ \alpha_i< \rho/8}} (1+\rho/4) \cdot \vert E_i \vert  \bigr)& \text{Using Corollary}~\ref{corol:alpha-frac}\\
    &\leq \frac{(1+\rho/4)^2\cdot (1+\beta)^2}{n} \cdot \sum_{i \in I} (1+\alpha_i) \cdot \vert E_i \vert &
\end{align*}
Where the last line is due to taking the max over values when $\alpha_i\geq \rho/8$, and $\alpha_i<\rho/8$. Similarly, we can show that 
\begin{align}
    \tilde{d} \geq \frac{(1-\rho/4)^2}{(1+\beta)n} \cdot \sum_{i \in I} (1+\alpha_i) \cdot \vert E_i \vert 
\end{align}
Using $\beta \leq \rho/8$,  
\begin{align*}
    &\tilde{d} \\&=\frac{(1 \pm (\rho/4))^2(1\pm \rho/8)^2}{n} \cdot \sum_{i \in I} (1+\alpha_i) \cdot \vert E_i \vert &\\
    &=\frac{1 \pm (3 \rho/2)}{n} \cdot \sum_{i \in I} (1+\alpha_i) \cdot \vert E_i \vert &\\
    &=\frac{1 \pm (3 \rho/2)}{n} \cdot \left( \sum_{i \in I} \vert E_i \vert +\sum_{i \in I} \alpha_i \cdot \vert E_i \vert\right) &\\
    &= \frac{1 \pm (3 \rho/2)}{n} \cdot \left( \sum_{i \in I} \vert E_i \setminus E'_i \vert + \sum_{i \in I} \vert E'_i \vert + \sum_{i \in I} \alpha_i \cdot \vert E_i \vert   \right) & \\
    &= \frac{1 \pm (3 \rho /2)}{n} \cdot \left( \sum_{i \in I} \vert E_i \setminus E'_i \vert + 2 \sum_{i\in I} \vert E'_i \vert \right) & \text{Since } \vert E'_i\vert = \alpha_i \cdot \vert E_i \vert \\
    &=\frac{1 \pm (3 \rho /2)}{n} \cdot \left( 2 \vert E(V \setminus U, V \setminus U)\vert + 2 \vert E(V \setminus U, U )\vert \right) & \text{Using }Equation~\ref{eq:sum-Ti} \text{ and }Equation~\ref{eq:sum-Ei-Ti}\\
    &=\frac{1 \pm (3 \rho /2)}{n} \cdot \left( 2 \vert E(V, V )\vert - 2 \vert E(U, U )\vert \right) &\\
\end{align*}

Where the last line is due to Corollary~\ref{corol:size-of-u}, which states that $\vert E( U,U) \vert <\frac{9}{25}\cdot {\rho n} + \frac{3\rho^{-11/4}}{2}\cdot {(2+1/\beta+\beta)\sqrt{\log(1+\beta)}} \cdot \sqrt{1+\frac{1}{\eps}}\cdot n^{3/4} \log^{9/4}(n)$ and by our assumption that $\bar{d} \geq 1$. Therefore,
\begin{align*}
\tilde{d} &=\bar{d}\left(1 \pm \frac{3 \rho}{2}\right)\cdot \left( 1\pm \left( \frac{9 \rho }{25} +  \frac{3\rho^{-11/4}}{2}\cdot {(2+1/\beta+\beta)\sqrt{\log(1+\beta)}} \cdot \sqrt{1+\frac{1}{\eps}}\cdot n^{-1/4} \log^{9/4}(n)\right) \right) \\
&=\bar{d}\left(1 \pm \frac{3 \rho}{2}\right)\cdot \left( 1\pm \left( \frac{9 \rho }{25} +  o(1)\right) \right)
\end{align*}

Since $\frac{93\rho}{50} + \frac{27 \rho^2}{50} + o(1) < 4\rho$, we have $\tilde{d}=\bar{d}(1\pm 4\rho)$. We can substitute $\rho$ by $\rho/4$ to obtain $\tilde{d} =\bar{d}(1\pm \rho)$. 
\end{proof}

\paragraph{CASE 2: $\vert S_1 \vert > 1.2T \cdot \sqrt{\vert S \vert}\cdot \vert S \vert $.}

Note that since $\vert S_1 \vert > 1.2T \cdot \sqrt{\vert S \vert}\cdot \vert S \vert $, the set of small buckets only consists of $U':= \{v \in \tilde{B}_i: (i \not \in I) \wedge (i>\log_{1+\beta} \left(\frac{6M_{\rho,n}}{\beta}\right)+2)\}$. Therefore, we redefine the set of edges between a noisy bucket and small buckets as $E'_i$, i.e., $E'_i := E(\tilde{B}_i,U') \subseteq E_i$, and $E'_1 := E({B}_1,U')$.

First, we show that the bucket $\vert B_1 \vert/n$ is now approximated well by $\vert S_1 \vert / \vert S \vert$ (see Part~\ref{it:big-approx-b1-s1} of Lemma~\ref{lem:big-approx}). We introduce a different estimator for counting edges between $B_1$ and small buckets given by $\frac{1}{\vert S \vert} (Z+\sum_{v \in S_1} (1+X(v)) \cdot \deg'(v) )$, where $Z\sim \Lap\left(36M_{\rho,n} \left(3+ \beta+\frac{1}{\beta}\right)\right)$ and $\deg'(v)= \min \{\deg(v),6M_{\rho,n} \left(3+ \beta+\frac{1}{\beta}\right)\}$, and the next few claims show that this gives an accurate approximation with high probability. The following lemma states that with high probability $\deg'(v)=\deg(v)$ for every $v\in S_1$ (See Algorithm~\ref{alg:noisy-avg-deg}).

\begin{lemma}\label{lem:deg'=degv}
For every $v\in S_1$, with probability at least $1-o(1)$, $\deg'(v)=\deg(v)$ where $\deg'(v)= \min \{\deg(v),6M_{\rho,n} \left(3+ \beta+\frac{1}{\beta}\right)\}$. 
\end{lemma}

\begin{proof}
Since $S_1= \cup_{i \leq \log_{1+\beta} \left(\frac{6M_{\rho,n}}{\beta}\right)+2 } \tilde{S}_i$, for every $i$ such that $\tilde{S}_i \subseteq S_1$, we have that $(1+\beta)^{i-2} \leq \frac{6M_{\rho,n}}{\beta}$. For every $v \in \tilde{S}_i$, we also know that $(1+\beta)^{i-1} \leq \tilde{d}(v) <(1+\beta)^i$. Therefore, 
\begin{align*}
    \tilde{d}(v) &<(1+\beta)^2 \cdot \frac{6M_{\rho,n}}{\beta} &\\
    \deg(v)+Y_v &< (1+\beta)^2 \cdot \frac{6M_{\rho,n}}{\beta} &\text{where }Y_v \sim \Lap(6/\eps) \\
    \deg(v)&< (1+\beta)^2 \cdot \frac{6M_{\rho,n}}{\beta} - Y_v&
\end{align*}
Using Corollary~\ref{corol:yi-small}, we have that with probability at least $1- \exp\left(\ln \lceil\log_{(1+\beta)}(n)\rceil - \frac{\log^{7/4} (n) }{3\rho^2} \right)$, 
\begin{align*}
    \deg(v)&< (1+\beta)^2 \cdot \frac{6M_{\rho,n}}{\beta} + 6M_{\rho,n}&\\
    &< 6M_{\rho,n} \left(3+ \beta+\frac{1}{\beta}\right)
\end{align*}
By a union bound, 
\begin{align*}
    &\Pr[\exists\ i\ :\ (\tilde{S}_i \subseteq S_1) \wedge (v \in \tilde{S}_i) \wedge \left(\deg(v) > \left(3+ \beta+\frac{1}{\beta}\right)6M_{\rho,n} \right)] \\
    &\leq \left(\log_{1+\beta} \left(\frac{6M_{\rho,n}}{\beta}\right)+2\right)  \exp\left(\ln \lceil\log_{(1+\beta)}(n)\rceil - \frac{\log^{7/4} (n) (1+\eps)}{3\rho^2} \right)\\
    &\leq \log_{1+\beta} \left( \frac{2 (1+1/\eps) \log^{7/4}(n) \cdot (1+\beta)^2}{\beta \cdot \rho^2}\right)\exp\left(\ln \lceil\log_{(1+\beta)}(n)\rceil - \frac{\log^{7/4} (n)(1+\eps) }{3\rho^2} \right)\\
    &\leq \exp\left(\ln \left(\log_{1+\beta} \left( \frac{2(1+1/\eps)\log^{7/4}(n) \cdot (1+\beta)^2}{\beta \cdot \rho^2}\right) \right) + \ln \lceil\log_{(1+\beta)}(n)\rceil - \frac{\log^{7/4} (n) (1+\eps)}{3\rho^2} \right)\\
\end{align*}

\end{proof}

Our main contribution in this case is Lemma~\ref{lem:case2-output-hoeffding} which shows that with high probability, our estimator (sans noise) approximates the fraction $(\vert E_1 \vert + \vert E'_1 \vert)/n$ quite well. 

\begin{proof}[Proof of Lemma~\ref{lem:case2-output-hoeffding}] \label{proof:lem-case2-output-hoeffding}
 Define the indicator random variable $Y(v) = 1$ if and only if $v \in S \cap B_1$.
\begin{align}
    &\frac{1}{\vert S \vert}\E \left[ \left(\sum_{v \in S_1} (1+X(v)) \cdot \deg(v)  \right) \right]\\
    &= \frac{1}{\vert S \vert}\E \left[ \left(\sum_{v \in S} Y(v) (1+X(v)) \cdot \deg(v)  \right) \right]\\
    &=\frac{1}{\vert S \vert}\left(\sum_{v \in S} \E[ \deg(v) (1+X(v)) | Y(v)=1] \Pr[Y(v)=1]  \right) \\
    &= \frac{\Pr[Y(v)=1] }{\vert S \vert}\left(\sum_{v \in S} \E[
    \deg(v) | Y(v)=1] + \sum_{v \in S} \E[\deg(v) X(v) | Y(v)=1] \right) \\
    &= \frac{\vert B_1 \vert}{\vert S \vert \cdot n}\left(\sum_{v \in S} \E[
    \deg(v) | Y(v)=1] + \sum_{v \in S} \E[\deg(v) X(v) | Y(v)=1] \right) \\
    &= \frac{\vert E_1 \vert }{n} +\frac{\vert B_1 \vert}{\vert S \vert \cdot n} \sum_{v \in S} \E[\deg(v) X(v) | Y(v)=1] \\
    &= \frac{\vert E_1 \vert }{n} +\frac{\vert B_1 \vert}{\vert S \vert \cdot n} \sum_{v \in S} \sum_{v \in B_1} \Pr[v \in B_1] \cdot \E[\deg(v) X(v) \vert v \in B_1] \\
    &= \frac{\vert E_1 \vert }{n} +\frac{\vert B_1 \vert}{\vert S \vert \cdot n} \sum_{v \in S} \sum_{v \in B_1}\frac{1}{\vert B_1 \vert}\cdot \deg(v) \E[X(v) \vert v \in B_1] \\
    &= \frac{\vert E_1 \vert }{n} +\frac{\vert B_1 \vert}{\vert S \vert \cdot n} \sum_{v \in S} \sum_{v \in B_1}\frac{1}{\vert B_1 \vert}\cdot \deg(v) \frac{\vert E'_1 \vert}{\deg(v)} \\
    &=\frac{\vert E_1 \vert + \vert E'_1 \vert}{n}
\end{align}
For every $v \in S$, define $W(v):= Y(v) (1+X(v)) \cdot \deg(v) $, then we can rewrite the result above as 
$$ \frac{1}{\vert S \vert}\E \left[ \sum_{v \in S} W(v)\right] = \frac{\vert E_1 \vert + \vert E'_1 \vert}{n}$$
Also, using the upper bound on $\deg(v)$ from Lemma~\ref{lem:deg'=degv}, $0 \leq W(v) <  12M_{\rho,n} \left(3+ \beta+\frac{1}{\beta}\right)$. 

\begin{enumerate}
    \item {\bfseries Case 1: $\bar{d}_1 \geq 1$.} By a multiplicative Hoeffding bound, 
\begin{align}
    &\Pr\left[\left\vert \frac{1}{\vert S \vert}\sum_{v \in S} W(v) - \frac{\vert E_1 \vert + \vert E'_1 \vert}{n}\right\vert \geq (\rho/4) \cdot \frac{\vert E_1 \vert + \vert E'_1 \vert}{n} \right]\\
    &\leq 2 \exp \left( - \frac{2 \vert S \vert^2 \cdot \frac{\rho^2}{16}\left(\frac{\vert E_1 \vert + \vert E'_1 \vert}{n} \right)^2}{\vert S \vert \cdot \left(12M_{\rho,n} \left(3+ \beta+\frac{1}{\beta}\right)\right)^2 }\right)
\end{align}
Observe that $\frac{\vert E_1 \vert + \vert E'_1 \vert}{n} \geq \frac{\vert B_1 \vert \cdot \bar{d}_1}{n} \geq \frac{\vert B_1 \vert}{n} $, also by our assumption, $\vert B_1 \vert >1.5 T\cdot \sqrt{\vert S \vert} \cdot n$, therefore, %$\frac{\vert E_1 \vert + \vert E'_1 \vert}{n} \geq  1.5T\cdot \sqrt{\vert S \vert}$. 
\begin{align}
    &2 \exp \left( - \frac{2 \vert S \vert \cdot \frac{\rho^2}{16}\left(\frac{\vert E_1 \vert + \vert E'_1 \vert}{n} \right)^2}{ \left(12M_{\rho,n} \left(3+ \beta+\frac{1}{\beta}\right)\right)^2 }\right)\\
    &\leq 2 \exp \left( - \frac{2 \vert S \vert \cdot \frac{\rho^2}{16}\left(1.5T \cdot \sqrt{\vert S \vert} \right)^2}{ \left(12M_{\rho,n} \left(3+ \beta+\frac{1}{\beta}\right)\right)^2 }\right)
\end{align}
Substituting the expressions for $ M_{\rho,n}$ and $T$, 

\begin{align}
    &2 \exp \left( - \frac{(9/32)\rho^2 \vert S \vert^2 \cdot  T^2 }{ \left(12M_{\rho,n} \left(3+ \beta+\frac{1}{\beta}\right)\right)^2 }\right)\\
    &=2 \exp \left( - \frac{(9/32)\rho^2 \vert S \vert^2 \cdot \left(\frac{1}{2}\sqrt{\frac{\rho}{n}}\cdot \frac{\eps}{1+\eps} \cdot \frac{1}{t} \right)^2}{ \left(12\left(\frac{1}{3} \cdot \sqrt{\frac{\rho}{n \sqrt{\log (n)}}} \cdot \frac{\vert S \vert }{t}\right)\left(3+ \beta+\frac{1}{\beta}\right)\right)^2 }\right)\\
    &=2 \exp \left( - \frac{9\rho^2}{2048(3+\beta+1/\beta)^2} \cdot \frac{\eps^2}{(1+\eps)^2} \cdot {\sqrt{\log (n)}} \right) 
\end{align}
\item {\bfseries Case 2: $\bar{d}_1 < 1$ and $\bar{d}\geq 1$.} By an additive Hoeffding bound,  
\begin{align}
    &\Pr\left[\left\vert \frac{1}{\vert S \vert}\sum_{v \in S} W(v) - \frac{\vert E_1 \vert + \vert E'_1 \vert}{n}\right\vert \geq \frac{\rho}{4} \right]\\
    &\leq 2 \exp \left( - \frac{2 \vert S \vert^2 \cdot \frac{\rho^2}{16}}{\vert S \vert \cdot \left(12M_{\rho,n} \left(3+ \beta+\frac{1}{\beta}\right)\right)^2 }\right)\\
    &\leq 2 \exp \left( - \frac{2 \vert S \vert \cdot \frac{\rho^2}{16}}{ \left(12M_{\rho,n} \left(3+ \beta+\frac{1}{\beta}\right)\right)^2 }\right)\\
     &\leq 2 \exp \left( - \frac{2 \vert S \vert \cdot \frac{\rho^2}{16}}{ \left(12M_{\rho,n} \left(3+ \beta+\frac{1}{\beta}\right)\right)^2 }\right)\\
     &\leq 2 \exp \left( - \frac{2 \left( t \cdot \frac{\log^2 (n)}{\rho^2} \cdot \sqrt{\frac{n}{\rho}} \cdot \left(1+\frac{1}{\eps}\right) \right)\cdot \frac{\rho^2}{16}}{  \left(12\left(\frac{1}{3} \cdot \sqrt{\frac{\rho}{n \sqrt{\log (n)}}} \cdot \frac{\vert S \vert }{t}\right)\left(3+ \beta+\frac{1}{\beta}\right)\right)^2}\right)\\
    &\leq  2 \exp \left( - \frac{\rho^{7/2}}{128 \log(1+\beta)} \cdot \frac{\eps}{1+\eps} \cdot \frac{\sqrt{n}}{\log^{1/2}(n)}\right)
\end{align}
\end{enumerate}
\end{proof}

The following claim about Laplace noise is used to show that the noise term $Z/\vert S \vert$ added to our estimator does not affect the accuracy by much. This is formally used in the main Lemma~\ref{lem:final-avg-deg-case2} for Case 2.  

\begin{claim}\label{clm:Z-lap-noise-small}
If $Z \sim \Lap\left(36M_{\rho,n} \left(3+ \beta+\frac{1}{\beta}\right)\right)$, then with probability at least $1-o(1)$, $\left\vert \frac{Z}{\vert S \vert}\right\vert < g(n)$ where $g(n):= \frac{4 (3+\beta + 1/\beta)(1+1/\eps)}{\rho^{3/2}} \cdot { \log(1+\beta)} \cdot \frac{\log^{1/2} (n)}{\sqrt{n}}=o_n(1)$.
\end{claim}
\begin{proof}
Using Lemma~\ref{lem:lap2}, with probability at least $1-o(1)$, we have 
\begin{align*}
\left\vert \frac{Z}{\vert S \vert} \right\vert &< \frac{36(3+\beta+1/\beta) M^2_{\rho,n}}{\vert S \vert} \\
&= \frac{36(3+\beta+1/\beta) \cdot \frac{\rho \cdot \vert S \vert^2}{9n \sqrt{\log n} \cdot t^2} }  {\vert S \vert} \\
&= \frac{4 (3+\beta + 1/\beta)(1+1/\eps)}{\rho^{3/2}} \cdot { \log(1+\beta)} \cdot \frac{\log^{1/2} (n)}{\sqrt{n}}
\end{align*} 

\end{proof}

We invoke the same lemmas as in Case 1 in the proof of the main Lemma~\ref{lem:final-avg-deg-case2} for Case 2 below to show that with high probability, the approximations of edges between the rest of the sufficiently large buckets, and between all the small buckets, as well as between the sufficiently large buckets and small buckets is good.

\begin{lemma}\label{lem:final-avg-deg-case2}
For every $\rho < 1/4$, $\beta \leq \rho/8$, and $\eps^{-1} =  o(\log^{1/4}(n))$, for sufficiently large $n$, and for the case when $\vert S_1 \vert > 1.2T \cdot \sqrt{\vert S \vert}\cdot \vert S \vert $, the main algorithm~(see Algorithm~\ref{alg:1+e}) outputs a value $\tilde{d}$ such that with probability at least $1-o(1)$, it holds that 
\[\left(1 -\rho \right)\cdot \bar{d} \leq \tilde{d} \leq \left(1 + \rho\right) \cdot \bar{d} \]
\end{lemma}

\begin{proof}
Note that the set of vertices that reside in noisy buckets deemed ``small'' by the sample, is defined by $U'= \{v \in \tilde{B}_i: (i \not \in I) \wedge (i>\log_{1+\beta} \frac{6M_{\rho,n}}{\beta}+2) \}$.

Also, 
\begin{align}
    \sum_{i \in I} \vert E'_i\vert + \vert E'_1 \vert= \vert E(V \setminus U', U') \vert \label{eq:sum-Ti} \\
    \sum_{i \in I} \vert E_i \setminus E'_i \vert + \vert E_1 \setminus E'_1\vert = 2 \vert E(V \setminus U', V \setminus U') \vert  \label{eq:sum-Ei-Ti} 
\end{align}
Let $\bar{d}_1$ be the average degree of bucket $B_1$. We do the analysis below assuming $\bar{d}_1 \geq 1$, and describe how the approximation factor changes when we assume $\bar{d}_1<1$, but $\bar{d}\geq 1$. With high probability, we have,
\begin{align*}
    \tilde{d} &= \frac{1}{\vert S \vert } \left(\sum_{i \in I } \vert \tilde{S}_i \vert \cdot (1+\tilde{\alpha}_i) \cdot (1+\beta)^{i} +Z+\sum_{v \in S_1} (1+X(v)) \cdot \deg'(v)  \right) &\\
    &\leq \frac{{(1+\rho/4)^2\cdot (1+\beta)^2}}{n} \cdot\left( \sum_{i \in I} \vert E_i \setminus E'_i \vert + 2 \sum_{i\in I} \vert E'_i \vert \right) +\frac{Z}{\vert S \vert} + \frac{1}{\vert S \vert}\sum_{v \in S_1} (1+X(v)) \cdot \deg'(v)  & \text{From Case 1 analysis}\\ 
    &= \frac{{(1+\rho/4)^2\cdot (1+\beta)^2}}{n}\left( \sum_{i \in I} \vert E_i \setminus E'_i \vert + 2 \sum_{i\in I} \vert E'_i \vert \right)  +\frac{Z}{\vert S \vert} + \frac{1}{\vert S \vert}\sum_{v \in S_1} (1+X(v)) \cdot \deg(v) &\text{From Lemma~\ref{lem:deg'=degv}} \\
    &\leq \frac{{(1+\rho/4)^2\cdot (1+\beta)^2}}{n}\left( \sum_{i \in I} \vert E_i \setminus E'_i \vert + 2 \sum_{i\in I} \vert E'_i \vert \right)  +\frac{Z}{\vert S \vert} + \frac{(1+\rho/4)}{n} \cdot (\vert E_1 \vert + \vert E'_1 \vert )  &\text{From Lemma~\ref{lem:case2-output-hoeffding}, Part~\ref{it:case2-output-hoeffding1}} \\
    &\leq \frac{{(1+\rho/4)^2\cdot (1+\beta)^2}}{n}\left( \sum_{i \in I} \vert E_i \setminus E'_i \vert + 2 \sum_{i\in I} \vert E'_i \vert \right)  +\frac{Z}{\vert S \vert} + \frac{(1+\rho/4)}{n} \cdot (\vert E_1\setminus E'_1 \vert + 2\vert E'_1 \vert )  & \\
    &\leq \frac{{(1+\rho/4)^2\cdot (1+\beta)^2}}{n} \cdot \left( \sum_{i \in I} \vert E_i \setminus E'_i \vert + \vert E_1 \setminus E'_1 \vert + 2 \left(\sum_{i\in I} \vert E'_i \vert + \vert E'_1 \vert\right) \right) + n^{-1/3}& \text{Using Claim~\ref{clm:Z-lap-noise-small}}\\
    &= \frac{{(1+\rho/4)^2\cdot (1+\beta)^2}}{n} \cdot \left( 2\vert E(V\setminus U', V \setminus U')\vert + 2\vert E(V\setminus U', U' \vert  \right) + n^{-1/3} & \\
    &= \frac{{(1+\rho/4)^2\cdot (1+\beta)^2}}{n} \cdot \left( 2 \vert E(V, V )\vert - 2 \vert E(U', U' )\vert + \frac{n^{2/3}}{(1+\rho/4)^2\cdot (1+\beta)^2}\right) & \\
\end{align*}
Similarly, we can show that with high probability,  
\begin{align*}
    \tilde{d} \geq \frac{(1-\rho/4)^2}{(1+\beta)n} \cdot \left( 2 \vert E(V, V )\vert - 2 \vert E(U', U' )\vert \right)
\end{align*}
Using $0<\beta \leq \rho/8$, 
\begin{align*}
    \tilde{d} &= \frac{1 \pm (3 \rho /2)}{n} \cdot \left( 2 \vert E(V, V )\vert - 2 \vert E(U', U' )\vert +\frac{n^{2/3}}{(1+\rho/4)^2} \right) &\\
    &=\frac{1 \pm (3 \rho/2)}{n} \cdot \left( \bar{d} n \pm \vert U'\vert^2 +\frac{n^{2/3}}{(1+\rho/4)^2}\right) &
\end{align*}
From  Lemma~\ref{lem:small-small}, Part~\ref{it:small-small-u'}, we know that $\vert U' \vert \leq \frac{3}{5}\cdot \sqrt{\rho n} $, and recall that we assume $\bar{d} \geq 1$, therefore, 
\[ \tilde{d} =\bar{d}\left(1 \pm \frac{3 \rho}{2}\right)\cdot \left( 1\pm \frac{9\rho }{25} + \frac{1}{(1+\rho/4)^2n^{1/3}}\right) \]
%\tanote{fix the expression below}
Since $\frac{93\rho}{50} + \frac{27 \rho^2}{50} + o(1) < 4\rho$, we have $\tilde{d}=\bar{d}(1\pm 4\rho)$. We can substitute $\rho$ by $\rho/4$ to obtain $\tilde{d} =\bar{d}(1\pm \rho)$.

When $\bar{d}_1<1$, but $\bar{d}\geq 1$, using Lemma~\ref{lem:case2-output-hoeffding}, Part~\ref{it:case2-output-hoeffding2}, and the same techniques as outlined above, we have, 
\[ \tilde{d} =\bar{d}\left(1 \pm \frac{3 \rho}{2}\right)\cdot \left( 1\pm \frac{9\rho }{25} + \frac{\rho/4}{(1+\rho/4)^2} + \frac{1}{(1+\rho/4)^2n^{1/3}}\right) \]
%\tanote{fix the expression below}
Since $\frac{93\rho}{50} + \frac{27 \rho^2}{50} + o(1) < 4\rho$, we have $\tilde{d}=\bar{d}(1\pm 4\rho)$. We can substitute $\rho$ by $\rho/4$ to obtain $\tilde{d} =\bar{d}(1\pm \rho)$. 
\end{proof}

\section{Proofs of Theorem~\ref{thm:main-sub-mm} and Theorem~\ref{thm:main-sub-vc}}\label{sec:main-sub-mm}
{\bf Notation.}
For ease of notation, when we are considering the Coupled Global Sensitivity of a graph algorithm with respect to edge-neighboring graphs, we denote it as $CGS^e$; and when we are considering the Coupled Global Sensitivity of a graph algorithm with respect to node-neighboring graphs, we denote it as $CGS^v$.

\subsection{The Maximal Matching and Vertex Cover Oracles}
We first describe the maximal matching oracle $\mathcal{O}^\pi_{MO}$ that is implemented recursively by~\cite{nguyen2008constant, yoshida2012improved, ORR12, behnezhad21} in Algorithm~\ref{alg:matching-oracle}. On input an edge $e$, Algorithm~\ref{alg:matching-oracle} queries all incident edges to $e$ of rank lower than $e$ to check if they belong to the matching $M$, while keeping track of which edges are in $M$ (defined greedily according to a fixed ranking $\pi$). Generating a random permutation $\pi\in Sym({n \choose 2})$ can be  simulated locally by assigning random values in the range $[0,1]$ to pairs of vertices of the graph at the moment when they are needed for the first time in the algorithm. To ensure that the rankings are distinct, one may employ a lazy sampling of the real numbers, see for e.g. Section 4.3 of~\cite{ORR12}. 
 In our case, in order to analyze subsequent algorithms that use $\cO^\pi_{MO}$ as a sub-routine, it will be enough to analyze the algorithm described in Algorithm~\ref{alg:matching-oracle} instead. In the sequel, we do not make any distinction between the  algorithm described in Algorithm~\ref{alg:matching-oracle} and oracle $\cO^{\pi}_{MO}$ that assigns rankings by sampling from $[0,1]$.
 
\begin{center}
\fbox{\parbox{\textwidth}{
\begin{enumerate}[nolistsep]
        \item \textbf{Input. }Given edge $e$, the oracle returns True if $e$ is in the Matching greedily created by the ranking $\pi$; returns False, otherwise.
\item If we have already computed $\mathcal{O}^\pi_{MO}(e)$, then return the computed answer.
\item Collect edges $e_1, \ldots, e_k$ sharing an endpoint with $e$ sorted by \emph{increasing rank}.
\item Initialize $i=1$. While $\pi(e_i) < \pi(e)$, if{$\mathcal{O}^\pi_{MO}(e_i)=\text{True}$ then return False, otherwise $i=i+1$.}
\item return True
    \end{enumerate} }}
    \captionof{algorithm}{Oracle $\mathcal{O}^\pi_{MO}(e)$ for a maximal matching based on ranking $\pi$ of edges.}  
\label{alg:matching-oracle}
\end{center}

The vertex cover oracle $\cO^\pi_{VC}$ (described in Algorithm~\ref{alg:VC-greedy-oracle}) is called on a vertex $v$ and subsequently calls the matching oracle on edges incident to $v$ to determine whether vertex $v$ is matched according to the ranking $\pi$. 

\begin{center}
\fbox{\parbox{\textwidth}{
 \textbf{Input. }Given vertex $v$, the oracle returns True if $v$ is in the Vertex Cover greedily created by the ranking; returns False otherwise.
\begin{enumerate}[nolistsep]
    \item Let $d_v = \deg(v)$.
    \item Collect edges $e_i = (v, \text{Nbr}(v,i))$ sorted by increasing rank $\forall i \in [d_v]$. 
    \item For {$i=1,\ldots,d_v$,} if {$\mathcal{O}^\pi_{MO}(e_i)=\text{True}$} then return True.
    \item Return False.
    \end{enumerate} }}
    \captionof{algorithm}{Oracle $\mathcal{O}^\pi_{VC}(v)$ for a vertex cover based on a randomly chosen ranking $\pi$ of edges.}  
\label{alg:VC-greedy-oracle}
\end{center}

\subsection{Formal Proofs of Theorem~\ref{thm:main-sub-mm} and Theorem~\ref{thm:main-sub-vc}}
\label{subsec:main-sub-mm}

We first state the non-private accuracy and time complexity guarantees for Algorithms~\ref{alg:MM-SLA-sampling} and \ref{alg:VC-SLA-sampling}.  
\begin{theorem}\cite{behnezhad21}
With probability $1-1/poly(n)$, $\cA_{sub-MM}$ and $\cA_{sub-VC}$ (Algorithms~\ref{alg:MM-SLA-sampling} and~\ref{alg:VC-SLA-sampling}) give a $(2,\rho n)$-approximation of maximum matching size and minimum vertex cover size respectively, with query/run time complexity $\tilde{O}\left((\bar{d}+1)/\rho^2 \right)$ where $\bar{d}$ denotes the average degree.
\end{theorem}

In what follows we analyze the CGS of the sampling algorithm for maximum matching denoted as $\cA_{sub-MM}$ (Algorithm~\ref{alg:MM-SLA-sampling}) with respect to node-neighboring graphs below. We also note that the CGS with respect to edge-neighboring graphs has the same upper bound and follows as a corollary. 

\begin{theorem}\label{thm:sub-MM-cgs}
\[CGS^v_{\cA_{sub\text{-}MM}} \leq \frac{n\rho^2}{16 \cdot 24 \ln n}\]
\end{theorem}
\begin{proof}

We can view the randomness $\mathcal{R} = \mathcal{R}_1 \times \mathcal{R}_2 $ as a joint-probability distribution. Here, $\mathcal{R}_1$ is the uniform distribution over $Sym({n \choose{2}})$ i.e., edge rankings. Similarly, $\mathcal{R}_2$ is the uniform distribution over ${ n \choose s}$ i.e., sets of $s$ vertices $v_1,\ldots, v_s$.

Let $G_1 \sim_v G_2$, and let $M_1$ and $M_2$ denote the respective maximal matchings computed greedily based on a fixed ranking $\pi \in Sym({n \choose{2}})$. Let $S_1$ denote the set of nodes that are endpoints of a matched edge in $M_1$, i.e., $S_1 = \{ u: \exists v ~s.t.~(u,v) \in M_1\}$; define $S_2$ analogously.

\begin{claim}\label{clm:vc-2}
$\left| |S_1| -|S_2| \right| \leq 2$
\end{claim}
\begin{proof}
Recall from the analysis for the CGS of the greedy algorithm $\cA_{MM}$ in Section~\ref{sec:tech-ov-match} that the matchings $M_1$ and $M_2$ differ in size by at most $1$, consequently, the claim follows.
\end{proof}

WLOG assume that $|S_1| \leq |S_2|$. Then we can always define a bijective function $f_\pi:[V] \rightarrow [V] $ with the following property: if $v \in S_1$ then $f(v)=v'$ corresponds to a vertex in $S_2$. Now we can define our permutation $\sigma: \mathcal{R} \rightarrow \mathcal{R}$ as follows: \[\sigma\left(\pi, \left\{v_1,\ldots, v_s\right\}\right) = \left(\pi, \left\{ f(v_1),\ldots, f(v_s) \right\} \right) \ . \]

Let $X^{(1)}_i$ equal 1 if $\cO^\pi_{VC}(v_i)$ returns True and 0 otherwise (see~Algorithm~\ref{alg:MM-SLA-sampling}) for the run of $\cA_{sub-MM}(G_1; \pi, \{v_j\}^s_{j=1})$. Similarly, define $X^{(2)}_i$ equal 1 if $\cO^\pi_{VC}(f(v_i))$ returns True and 0 otherwise for the run of $\cA_{sub-MM}(G_2; \sigma(\pi, \{v_j\}^s_{j=1}))$.

Since we sample without replacement we have $\left|\sum_{i \in [s]} X_i^{(1)} - \sum_{i \in [s]} X_i^{(2)} \right| \leq \left| S_2 \setminus S_1\right| \leq 2$, where the last inequality is by Claim~\ref{clm:vc-2}. Thus, 
\begin{align*}
\left\vert    \frac{n}{2s}(\sum_{i \in [s']}X^{(1)}_i) -\frac{n}{2s}(\sum_{i \in [s']}X^{(2)}_i) \right\vert \leq \frac{2n}{2s} \leq \frac{n\rho^2}{16 \cdot 24 \ln n} \;,
\end{align*}
where the last inequality comes from substituting the value for sample size $s$. 
\end{proof}

\begin{corollary}[Differentially-private $\cA_{sub-MM}$] \label{corol:dp-MM-SLA}
Let $\cA_{sub-MM}(G)$ be as described in Algorithm~\ref{alg:MM-SLA-sampling}. Then the algorithm $\cA^{DP}_{sub-MM}(G):= \cA_{sub-MM}(G)+ \Lap\left(\frac{n\rho^2}{16 \cdot 24 \cdot \eps \ln n}\right)$ is $\eps$-node (and edge) differentially private. %\enote{a bit of a weird notation}
\end{corollary}
\begin{proof}
This follows from Theorem~\ref{thm:lap-cgs} and Theorem~\ref{thm:sub-MM-cgs}.
\end{proof}

The following claim gives an accuracy guarantee for $\cA^{DP}_{sub-MM}(G)$.

\begin{claim}\label{clm:max-match-claim} [Accuracy of $\cA^{DP}_{sub-MM}(G)$]
Let $\pi$ be an arbitrary ranking, and let $M$ be the size of the maximum matching in $G$. Let $\tilde{M} := \cA_{sub-MM}(G)$. Then with probability $1- (2/n^4 + 1/n^{\frac{192 \cdot \eps}{\rho}})$, it is the case that 
\[  \frac{M}{2}  - \frac{3\rho n}{2} \leq \tilde{M} + \Lap\left(\frac{n\rho^2}{16 \cdot 24 \cdot \eps \ln n}\right)  \leq M  + \frac{\rho n}{2} \]
for some $\rho>0$, where $\eps$ is the privacy parameter.% and $d$ is the maximum degree of the graph.
\end{claim} 

\begin{proof}
The correctness analysis is identical to that in~\cite{behnezhad21}. The only difference in our algorithm is that we sample without replacement, but using Fact~\ref{fact:conc}, we can use the same concentration bounds as used when sampling with replacement.

We repeat the argument here for the sake of completeness. Let $M_\pi$ be the resulting maximal matching when edges are chosen greedily according to the ranking $\pi$. From the definition of Algorithm~\ref{alg:MM-SLA-sampling}, $X_i=1$ if and only if the vertex $v_i$ (sampled without replacement) for random permutation $\pi$ is matched in $M_\pi$. Thus, we have 
$ \E[X_i ] = \Pr[X_i=1] = \frac{2 \E[ \vert M_\pi \vert ]}{n}$. Let $X= \sum_{i \in s} X_i $ and $\E[X] = \frac{2s \E[ \vert M_\pi \vert ]}{n} $. 

Using a Chernoff bound and Fact~\ref{fact:conc}, 
$$ \Pr[\vert X - \E X \vert \geq \sqrt{12 \E X \ln n} ] \leq 2 \exp \left( - \frac{12 \E X \ln n}{3 \E X}\right) = 2/n^4$$
Now, with probability $1-2/n^4$ we have,
\begin{align*}
    \frac{X n }{2s} &\in \frac{1}{2} \cdot \frac{(\E [X] \pm \sqrt{12 \E [X] \ln n} )n}{s}\\
    &= \frac{1}{2} \cdot \left( \frac{n \E [X]}{s} \pm \frac{\sqrt{12 \E [X] n^2 s^{-2}\ln n} }{s} \right) \\ 
    &= \E \vert M_\pi \vert \pm \frac{1}{2} \cdot \sqrt{24 \E [\vert M_\pi \vert ] n s^{-1}\ln n} \\
    &= \E \vert M_\pi \vert \pm \frac{1}{2} \cdot \sqrt{\E[\vert M_\pi \vert ] \rho^2 n/ 16} \\
    &\in \E \vert M_\pi \vert \pm \rho n/8
\end{align*}
where the last step is because $\E \vert M_\pi \vert \leq n$. Since our estimator is $\frac{Xn}{2s} - \frac{\rho n}{2}$ and we know that $\frac{1}{2} M \leq \E \vert M_\pi \vert \leq M$, we have that with probability $1-2/n^4$, $ M/2 - \rho n \leq \tilde{M} \leq  M$. 

Using Fact~\ref{fact:lap}, we have, 
$$\Pr \left[ \left\vert \Lap\left(\frac{n\rho^2}{16 \cdot 24 \cdot \eps \ln n}\right)\right\vert \geq \frac{\rho n}{2} \right] \leq \exp \left(- \frac{8 \cdot 24 \cdot \eps \ln n}{\rho}\right) \;.$$
%Since we consider $\eps,\rho$ are also chosen to be constant, the RHS is $1/poly(n)$, and our claim follows. 

Thus with probability $1- (2/n^4 + 1/n^{\frac{192 \cdot \eps}{\rho}})$, our claim follows. 
\end{proof}

Observe that by subtracting $\frac{\rho n}{2}$ from $\tilde{M} + \Lap\left(\frac{n\rho^2}{16 \cdot 24 \cdot \eps \ln n}\right)$  we can ensure that our estimate lies in the range $\left[ M/2  - 2\rho n ,  M\right]$ with probability $1- (2/n^4 + 1/n^{\frac{192 \cdot \eps}{\rho}})$.

\begin{proof}[Proof of Theorem~\ref{thm:main-sub-mm}]
The query/time complexity analysis of Algorithm \ref{alg:MM-SLA-sampling} follows from~\cite{behnezhad21}. The privacy guarantee follows from Corollary~\ref{corol:dp-MM-SLA}. The accuracy guarantee follows from Claim~\ref{clm:max-match-claim}.
\end{proof}

We now describe the sampling algorithm that estimates the minimum vertex cover size (see Algorithm~\ref{alg:VC-SLA-sampling}) which is identical to Algorithm~\ref{alg:MM-SLA-sampling} except it returns a different estimator.

\begin{center}
\fbox{\parbox{\textwidth}{
  \textbf{Input. }Input Graph $G=(V,E)$.
\begin{enumerate}[nolistsep]
    \item Uniformly and independently sample $s=16 \cdot 24 \ln n/\rho^2$ vertices from $V$ without replacement.
    \item For{ $i=1 \ldots s$,} if {$\mathcal{O}^\pi_{VC}(v_i)=$ True} then let $X_i = 1$, otherwise let $X_i=0$. 
    \item Let $\tilde{C}=\frac{n}{s}(\sum_{i \in [s]}X_i)+\frac{\rho n}{4}$.
    \end{enumerate} }}
    \captionof{algorithm}{Local Vertex Cover algorithm $\mathcal{A}_{sub-MM-VC}$ using Oracle access.}  
\label{alg:VC-SLA-sampling}
\end{center}

We analyze the CGS of $\cA_{sub-VC}$ with respect to node-neighboring graphs below. We also note that the CGS with respect to edge-neighboring graphs has the same upper bound and follows as a corollary. 

\begin{theorem}\label{thm:sub-VC-cgs}
\[CGS^v_{\cA_{sub\text{-}VC}} \leq  \frac{n \rho^2 }{192 \ln n}\]
\end{theorem}
\begin{proof}
The proof is identical to Theorem~\ref{thm:sub-MM-cgs}, except now we are accounting for a different estimator. We include the entire proof for the sake of completeness. 

As before, we can view the randomness $\mathcal{R} = \mathcal{R}_1 \times \mathcal{R}_2 $ as a joint-probability distribution. Here, $\mathcal{R}_1$ is the uniform distribution over $Sym({n \choose{2}})$ i.e., edge rankings. Similarly, $\mathcal{R}_2$ is the uniform distribution over ${ n \choose s}$ i.e., sets of $s$ vertices $v_1,\ldots, v_s$.

Let $G_1 \sim_v G_2$, and let $M_1$ and $M_2$ denote the respective maximal matchings computed greedily based on a fixed ranking $\pi \in Sym({n \choose{2}})$. Let $S_1$ denote the set of nodes that are endpoints of a matched edge in $M_1$, i.e., $S_1 = \{ u: \exists v ~s.t.~(u,v) \in M_1\}$; define $S_2$ analogously.

WLOG assume that $|S_1| \leq |S_2|$. Then we can always define a bijective function $f_\pi:[V] \rightarrow [V] $ with the following property: if $v \in S_1$ then $f(v)=v'$ corresponds to a vertex in $S_2$. Now we can define our permutation $\sigma: \mathcal{R} \rightarrow \mathcal{R}$ as follows: \[\sigma\left(\pi, \left\{v_1,\ldots, v_s\right\}\right) = \left(\pi, \left\{ f(v_1),\ldots, f(v_s) \right\} \right) \ . \]

Let $X^{(1)}_i$ equal 1 if $\cO^\pi_{VC}(v_i)$ returns True and 0 otherwise (see~Algorithm~\ref{alg:VC-SLA-sampling}) for the run of $\cA_{sub-VC}(G_1; \pi, \{v_j\}^s_{j=1})$. Similarly, define $X^{(2)}_i$ equal 1 if $\cO^\pi_{VC}(f(v_i))$ returns True and 0 otherwise for the run of $\cA_{sub-VC}(G_2; \sigma(\pi, \{v_j\}^s_{j=1}))$. 

%Observe that $X_i^{(1)} = X_i^{(2)}$ as long as $v_i \not \in S_2 \setminus S_1$. % --- $X_i^1 = 1$ if and only if $v_i \in S_1$. 
Since we sample without replacement we have $\left|\sum_{i \in [s]} X_i^{(1)} - \sum_{i \in [s]} X_i^{(2)} \right| \leq \left| S_2 \setminus S_1\right| \leq 2$, where the last inequality is by Claim~\ref{clm:vc-2}. Thus, 
\begin{align*}
\left\vert    \frac{n}{s}(\sum_{i \in [s]}X^{(1)}_i) -\frac{n}{s}(\sum_{i \in [s]}X^{(2)}_i) \right\vert \leq \frac{2n}{s} \leq \frac{2n\rho^2}{16 \cdot 24 \ln n}= \frac{n \rho^2 }{8 \cdot 24\ln n}\;.
\end{align*}

\end{proof}

\begin{corollary}[Differentially-private $\cA_{sub-VC}$] \label{corol:dp-VC-SLA}
Let $\cA_{sub-VC}(G)$ be as described in Algorithm~\ref{alg:VC-SLA-sampling}, then $\cA^{DP}_{sub-VC}(G):= \cA_{sub-VC}(G)+ \Lap\left({n\rho^2}/(8 \cdot 24 \eps \ln n) \right)$ is $\eps$-node (and edge) differentially private. 
\end{corollary}
\begin{proof}
This follows from Theorem~\ref{thm:lap-cgs} and Theorem~\ref{thm:sub-VC-cgs}.
\end{proof}

The following claim gives an accuracy guarantee for $\cA^{DP}_{sub-VC}(G)$. Given that $C$ denotes the minimum vertex cover size and $\tilde{C}$ is the vertex cover whose size is estimated by $\cA_{sub-VC}(G)$ (as stated below), 
\begin{claim}\label{clm:min-vc-claim} [Accuracy of $\cA^{DP}_{sub-VC}(G)$]
Let $\pi$ be a fixed ranking on the (existing and non-existing) edges of $G$, and let $C$ denote the size of the minimum vertex cover. Let $\tilde{C} := \cA_{sub-VC}(G)$. Then with probability $1-(2/n^4 + 1/n^{96 \eps/\rho } )$, 
\[  C - \frac{\rho n}{2} \leq \tilde{C} + \Lap\left(\frac{n\rho^2}{8 \cdot 24\eps \ln n}\right)  \leq 2 C  + \frac{3\rho n}{2} \]
for some $\rho>0$, where $\eps$ is the privacy parameter.
\end{claim}
\begin{proof}
The correctness analysis is identical to that in~\cite{behnezhad21}. The only difference in our algorithm is that we sample without replacement, but using Fact~\ref{fact:conc}, we can use the same concentration bounds as used when sampling with replacement. From~\cite{behnezhad21}, we know that with probability $1-2/n^4$, $C  \leq \tilde{C} \leq 2  C  + \rho n$.

And using Fact~\ref{fact:lap}, we have, 
$$\Pr \left[ \left\vert \Lap\left(\frac{n\rho^2}{8 \cdot 24\eps \ln n}\right) \right\vert \geq \frac{\rho n}{2} \right] \leq \exp \left(- \frac{4 \cdot 24 \cdot \eps \ln n}{\rho}\right) \;.$$
Therefore with probability $1-(2/n^4 + 1/n^{96 \eps/\rho } )$ our claim follows.
\end{proof}

We observe that we can always add $\frac{\rho n}{2}$ to $ \tilde{C} +\Lap\left(\frac{n\rho^2}{8 \cdot 24\eps \ln n}\right) $ to ensure that our estimate lies in the range $\left[ C , 2 C + 2 \rho n \right]$ with probability $1-(2/n^4 + 1/n^{96 \eps/\rho } )$.

\begin{proof}[Proof of Theorem~\ref{thm:main-sub-vc}]
The query/time complexity analysis of Algorithm~\ref{alg:VC-SLA-sampling} follows from~\cite{behnezhad21}. The privacy guarantee follows from Corollary~\ref{corol:dp-VC-SLA}. The accuracy guarantee follows from Claim~\ref{clm:min-vc-claim}.
\end{proof}

\section{Conclusions and open questions}\label{sec:open}

In this work we give a differentially-private sublinear-time $(1+\rho)$-approximation algorithm for estimating the average degree of the graph. We achieve a running time comparable to its non-private counterpart, which is also tight in terms of its asymptotic behaviour with respect to the number of vertices of the graph. We also give the first differentially-private approximation algorithms for the problems of estimating maximum matching size and vertex cover size of a graph. 

To analyze the privacy of our algorithms, we proposed the notion of coupled global sensitivity, as a generalization of global sensitivity, which is applicable to randomized approximation algorithms. We show that coupled global sensitivity implies differential privacy, and use it to show that previous non-private algorithms from the literature, or  variants, can be made private by finely tuning the amounts of noise added in various steps of the algorithms.

We propose several directions of investigation for developing the notion of coupled global sensivity further and open problems pertaining to differentially-private sublinear-time algorithms for graphs.

{\bf Other applications and limitations of CGS} In particular, what are the limitations of the CGS method? Can we characterize the set of algorithms with small  CGS? Are there other natural problems for which we already have algorithms with small CGS, and hence that are easily amenable to privacy analogues? Are there algorithms for which we can prove large lower bounds on the CGS and yet they provide differential privacy?  

{\bf Better approximations for  maximum matching problems}  
In~\cite{nguyen2008constant,yoshida2012improved}, the authors also give a $(1,\rho n)$-approximation of maximum matching size with a query complexity that is exponential in $d$. Their analysis involves iterating over a sequence of oracles to augment paths of small length, in increasing order of lengths. The matching oracle considered in this work is used only in the first iteration. Analyzing the coupled global sensitivity of that algorithm appears to be much more involved, and we leave it as an open problem.

{\bfseries Better time complexity guarantees for $(2,\rho n)$-approximation matching and vertex cover algorithms. } Note that our results in Theorems~\ref{thm:main-sub-mm} and~\ref{thm:main-sub-vc} achieve an expected running time. In contrast, the results  in~\cite{behnezhad21} achieve a high-probability bound on the time-complexity. This can be done by running multiple instances of the resulting approximation algorithm for enough time and returning the output of the instance that terminates first (the analysis involves a simple application of Markov inequality). Achieving this step in a way that preserves privacy would result in a degradation of the privacy parameter $\eps$, due to composition. We leave it as an open question to provide a tighter privacy vs time-complexity analysis.

% {\bf Improved query complexity for vertex cover problems}
% In~\cite{ORR12}, the authors also give a query complexity result in terms of average degree of the graph instead of the maximum degree of the graph. The transformation applied to the input graph in this case automatically adds high-degree vertices to the cover and proceeds by finding a cover for the graph that is induced by the remaining vertices. Unfortunately, this transformation is highly sensitive and adding noise proportional to coupled global sensitivity is not useful here. Instead, one way to preserve privacy is to add noise individually to the threshold used for choosing high vs low degree vertices in the input graph. We leave this as an open problem. 

%\newpage

%\noindent {\bf Acknowledgements}
\section{Acknowledgements} 
We thank several anonymous reviewers for their valuable feedback on a preliminary version of this work. We thank Soheil Behnezhad for bringing to our attention the subtlety in the analysis of \cite{ORR12}, first discovered by \cite{ChenKK20}, and finally resolved in his recent work \cite{behnezhad21}.

\bibliographystyle{plain}
\bibliography{references}

\begin{thebibliography}{10}

\bibitem{alabi2020differentially}
Daniel Alabi, Audra McMillan, Jayshree Sarathy, Adam~D. Smith, and Salil~P.
  Vadhan.
\newblock Differentially private simple linear regression.
\newblock {\em CoRR}, abs/2007.05157, 2020.

\bibitem{bardenet2015concentration}
R{\'e}mi Bardenet and Odalric-Ambrym Maillard.
\newblock Concentration inequalities for sampling without replacement.
\newblock {\em Bernoulli}, 21(3):1361--1385, 2015.

\bibitem{behnezhad21}
Soheil Behnezhad.
\newblock Time-optimal sublinear algorithms for matching and vertex cover.
\newblock {\em CoRR}, abs/2106.02942, 2021.

\bibitem{blocki2013differentially}
Jeremiah Blocki, Avrim Blum, Anupam Datta, and Or~Sheffet.
\newblock Differentially private data analysis of social networks via
  restricted sensitivity.
\newblock In Robert~D. Kleinberg, editor, {\em Innovations in Theoretical
  Computer Science, {ITCS} '13, Berkeley, CA, USA, January 9-12, 2013}, pages
  87--96. {ACM}, 2013.

\bibitem{BCS2015}
Christian Borgs, Jennifer~T. Chayes, and Adam~D. Smith.
\newblock Private graphon estimation for sparse graphs.
\newblock In Corinna Cortes, Neil~D. Lawrence, Daniel~D. Lee, Masashi Sugiyama,
  and Roman Garnett, editors, {\em Advances in Neural Information Processing
  Systems 28: Annual Conference on Neural Information Processing Systems 2015,
  December 7-12, 2015, Montreal, Quebec, Canada}, pages 1369--1377, 2015.

\bibitem{Borgs_2018}
Christian Borgs, Jennifer~T. Chayes, Adam~D. Smith, and Ilias Zadik.
\newblock Revealing network structure, confidentially: Improved rates for
  node-private graphon estimation.
\newblock In Mikkel Thorup, editor, {\em 59th {IEEE} Annual Symposium on
  Foundations of Computer Science, {FOCS} 2018, Paris, France, October 7-9,
  2018}, pages 533--543. {IEEE} Computer Society, 2018.

\bibitem{CV_NIPS2013}
Kamalika Chaudhuri and Staal~A. Vinterbo.
\newblock A stability-based validation procedure for differentially private
  machine learning.
\newblock In Christopher J.~C. Burges, L{\'{e}}on Bottou, Zoubin Ghahramani,
  and Kilian~Q. Weinberger, editors, {\em Advances in Neural Information
  Processing Systems 26: 27th Annual Conference on Neural Information
  Processing Systems 2013. Proceedings of a meeting held December 5-8, 2013,
  Lake Tahoe, Nevada, United States}, pages 2652--2660, 2013.

\bibitem{Chen_2013}
Shixi Chen and Shuigeng Zhou.
\newblock Recursive mechanism: towards node differential privacy and
  unrestricted joins.
\newblock In Kenneth~A. Ross, Divesh Srivastava, and Dimitris Papadias,
  editors, {\em Proceedings of the {ACM} {SIGMOD} International Conference on
  Management of Data, {SIGMOD} 2013, New York, NY, USA, June 22-27, 2013},
  pages 653--664. {ACM}, 2013.

\bibitem{ChenKK20}
Yu~Chen, Sampath Kannan, and Sanjeev Khanna.
\newblock Sublinear algorithms and lower bounds for metric {TSP} cost
  estimation.
\newblock In Artur Czumaj, Anuj Dawar, and Emanuela Merelli, editors, {\em 47th
  International Colloquium on Automata, Languages, and Programming, {ICALP}
  2020, July 8-11, 2020, Saarbr{\"{u}}cken, Germany (Virtual Conference)},
  volume 168 of {\em LIPIcs}, pages 30:1--30:19. Schloss Dagstuhl -
  Leibniz-Zentrum f{\"{u}}r Informatik, 2020.

\bibitem{DKS14}
Anirban Dasgupta, Ravi Kumar, and Tam{\'{a}}s Sarl{\'{o}}s.
\newblock On estimating the average degree.
\newblock In Chin{-}Wan Chung, Andrei~Z. Broder, Kyuseok Shim, and Torsten
  Suel, editors, {\em 23rd International World Wide Web Conference, {WWW} '14,
  Seoul, Republic of Korea, April 7-11, 2014}, pages 795--806. {ACM}, 2014.

\bibitem{Dwork_McSherry_Nissim_Smith_2017}
Cynthia Dwork, Frank McSherry, Kobbi Nissim, and Adam Smith.
\newblock Calibrating noise to sensitivity in private data analysis.
\newblock {\em Journal of Privacy and Confidentiality}, 7(3):17–51, May 2017.

\bibitem{DR14}
Cynthia Dwork and Aaron Roth.
\newblock The algorithmic foundations of differential privacy.
\newblock {\em Found. Trends Theor. Comput. Sci.}, 9(3-4):211--407, 2014.

\bibitem{feige2006sums}
Uriel Feige.
\newblock On sums of independent random variables with unbounded variance and
  estimating the average degree in a graph.
\newblock {\em {SIAM} J. Comput.}, 35(4):964--984, 2006.

\bibitem{fichtenberger2021differentially}
Hendrik Fichtenberger, Monika Henzinger, and Wolfgang Ost.
\newblock Differentially private algorithms for graphs under continual
  observation.
\newblock In Petra Mutzel, Rasmus Pagh, and Grzegorz Herman, editors, {\em 29th
  Annual European Symposium on Algorithms, {ESA} 2021, September 6-8, 2021,
  Lisbon, Portugal (Virtual Conference)}, volume 204 of {\em LIPIcs}, pages
  42:1--42:16. Schloss Dagstuhl - Leibniz-Zentrum f{\"{u}}r Informatik, 2021.

\bibitem{gehrke2011towards}
Johannes Gehrke, Edward Lui, and Rafael Pass.
\newblock Towards privacy for social networks: {A} zero-knowledge based
  definition of privacy.
\newblock In Yuval Ishai, editor, {\em Theory of Cryptography - 8th Theory of
  Cryptography Conference, {TCC} 2011, Providence, RI, USA, March 28-30, 2011.
  Proceedings}, volume 6597 of {\em Lecture Notes in Computer Science}, pages
  432--449. Springer, 2011.

\bibitem{goldreich2004estimating}
Oded Goldreich and Dana Ron.
\newblock Approximating average parameters of graphs.
\newblock {\em Random Struct. Algorithms}, 32(4):473--493, 2008.

\bibitem{gupta2010differentially}
Anupam Gupta, Katrina Ligett, Frank McSherry, Aaron Roth, and Kunal Talwar.
\newblock Differentially private combinatorial optimization.
\newblock In Moses Charikar, editor, {\em Proceedings of the Twenty-First
  Annual {ACM-SIAM} Symposium on Discrete Algorithms, {SODA} 2010, Austin,
  Texas, USA, January 17-19, 2010}, pages 1106--1125. {SIAM}, 2010.

\bibitem{hay2009accurate}
Michael Hay, Chao Li, Gerome Miklau, and David~D. Jensen.
\newblock Accurate estimation of the degree distribution of private networks.
\newblock In Wei Wang, Hillol Kargupta, Sanjay Ranka, Philip~S. Yu, and Xindong
  Wu, editors, {\em {ICDM} 2009, The Ninth {IEEE} International Conference on
  Data Mining, Miami, Florida, USA, 6-9 December 2009}, pages 169--178. {IEEE}
  Computer Society, 2009.

\bibitem{hoeffding1994probability}
Wassily Hoeffding.
\newblock Probability inequalities for sums of bounded random variables.
\newblock In {\em The collected works of Wassily Hoeffding}, pages 409--426.
  Springer, 1994.

\bibitem{karwa2011private}
Vishesh Karwa, Sofya Raskhodnikova, Adam~D. Smith, and Grigory Yaroslavtsev.
\newblock Private analysis of graph structure.
\newblock {\em {ACM} Trans. Database Syst.}, 39(3):22:1--22:33, 2014.

\bibitem{kasiviswanathan2013analyzing}
Shiva~Prasad Kasiviswanathan, Kobbi Nissim, Sofya Raskhodnikova, and Adam~D.
  Smith.
\newblock Analyzing graphs with node differential privacy.
\newblock In Amit Sahai, editor, {\em Theory of Cryptography - 10th Theory of
  Cryptography Conference, {TCC} 2013, Tokyo, Japan, March 3-6, 2013.
  Proceedings}, volume 7785 of {\em Lecture Notes in Computer Science}, pages
  457--476. Springer, 2013.

\bibitem{LG2014}
Wentian Lu and Gerome Miklau.
\newblock Exponential random graph estimation under differential privacy.
\newblock In Sofus~A. Macskassy, Claudia Perlich, Jure Leskovec, Wei Wang, and
  Rayid Ghani, editors, {\em The 20th {ACM} {SIGKDD} International Conference
  on Knowledge Discovery and Data Mining, {KDD} '14, New York, NY, {USA} -
  August 24 - 27, 2014}, pages 921--930. {ACM}, 2014.

\bibitem{nguyen2008constant}
Huy~N. Nguyen and Krzysztof Onak.
\newblock Constant-time approximation algorithms via local improvements.
\newblock In {\em 49th Annual {IEEE} Symposium on Foundations of Computer
  Science, {FOCS} 2008, October 25-28, 2008, Philadelphia, PA, {USA}}, pages
  327--336. {IEEE} Computer Society, 2008.

\bibitem{NRS07}
Kobbi Nissim, Sofya Raskhodnikova, and Adam~D. Smith.
\newblock Smooth sensitivity and sampling in private data analysis.
\newblock In David~S. Johnson and Uriel Feige, editors, {\em Proceedings of the
  39th Annual {ACM} Symposium on Theory of Computing, San Diego, California,
  USA, June 11-13, 2007}, pages 75--84. {ACM}, 2007.

\bibitem{ORR12}
Krzysztof Onak, Dana Ron, Michal Rosen, and Ronitt Rubinfeld.
\newblock A near-optimal sublinear-time algorithm for approximating the minimum
  vertex cover size.
\newblock In Yuval Rabani, editor, {\em Proceedings of the Twenty-Third Annual
  {ACM-SIAM} Symposium on Discrete Algorithms, {SODA} 2012, Kyoto, Japan,
  January 17-19, 2012}, pages 1123--1131. {SIAM}, 2012.

\bibitem{parnas2007approximating}
Michal Parnas and Dana Ron.
\newblock Approximating the minimum vertex cover in sublinear time and a
  connection to distributed algorithms.
\newblock {\em Theor. Comput. Sci.}, 381(1-3):183--196, 2007.

\bibitem{raskhodnikova2015efficient}
Sofya Raskhodnikova and Adam~D. Smith.
\newblock Efficient lipschitz extensions for high-dimensional graph statistics
  and node private degree distributions.
\newblock {\em CoRR}, abs/1504.07912, 2015.

\bibitem{Ron19}
Dana Ron.
\newblock Sublinear-time algorithms for approximating graph parameters.
\newblock In {\em Computing and Software Science}, volume 10000 of {\em Lecture
  Notes in Computer Science}, pages 105--122. Springer, 2019.

\bibitem{sealfon2019efficiently}
Adam Sealfon and Jonathan~R. Ullman.
\newblock Efficiently estimating erdos-renyi graphs with node differential
  privacy.
\newblock {\em J. Priv. Confidentiality}, 11(1), 2021.

\bibitem{seshadhri2015simpler}
C.~Seshadhri.
\newblock A simpler sublinear algorithm for approximating the triangle count.
\newblock {\em CoRR}, abs/1505.01927, 2015.

\bibitem{sivasubramaniamdifferentially}
Harry Sivasubramaniam, Haonan Li, and Xi~He.
\newblock Differentially private sublinear average degree approximation.

\bibitem{SLMVC18}
Shuang Song, Susan Little, Sanjay Mehta, Staal~A. Vinterbo, and Kamalika
  Chaudhuri.
\newblock Differentially private continual release of graph statistics.
\newblock {\em CoRR}, abs/1809.02575, 2018.

\bibitem{west2001introduction}
Douglas~Brent West et~al.
\newblock {\em Introduction to graph theory}, volume~2.
\newblock Prentice hall Upper Saddle River, 2001.

\bibitem{yoshida2012improved}
Yuichi Yoshida, Masaki Yamamoto, and Hiro Ito.
\newblock Improved constant-time approximation algorithms for maximum matchings
  and other optimization problems.
\newblock {\em {SIAM} J. Comput.}, 41(4):1074--1093, 2012.

\bibitem{ZCP15}
Jun Zhang, Graham Cormode, Cecilia~M. Procopiuc, Divesh Srivastava, and Xiaokui
  Xiao.
\newblock Private release of graph statistics using ladder functions.
\newblock In Timos~K. Sellis, Susan~B. Davidson, and Zachary~G. Ives, editors,
  {\em Proceedings of the 2015 {ACM} {SIGMOD} International Conference on
  Management of Data, Melbourne, Victoria, Australia, May 31 - June 4, 2015},
  pages 731--745. {ACM}, 2015.

\end{thebibliography}
 
\newpage 
\appendix
%\section{Appendix}
\section{A Simple Example of CGS}\label{sec:cgs-example}
\paragraph{Example.} As a simple motivating example, suppose we have access to records of individuals in the form of their name and profession with entries sorted in lexicographic order (by name). Consider the function $f(D):=${number of doctors in dataset }$D$, along with the following approximation algorithm $\cA_f(D)$: (1) Sample each record of $D$ (with probability 1/2) without replacement. Let $S$ denote the resulting sample.  (2) Return $f(S)$. We can make $\cA_f(D)$ differentially private by adding noise proportional to $CGS_{\cA_f}$ (see Theorem~\ref{thm:lap-cgs}). For accuracy purposes, we need to show that $CGS_{\cA_f}$ is small. Observe that the set of random coin tosses $\cR$ is defined over the sampling procedure itself, i.e., if heads, $\cA_f$ includes the record in the sample; otherwise, it does not. Let $D_1, D_2$ be two neighboring datasets i.e., we can find $d^*_1 \in D_1 $ and $d^*_2 \in D_2 $ such that $D_1 \setminus \{d_1^*\} = D_2 \setminus \{d_2^*\}$. 

We can argue that $CGS_{\cA_f}$ is at most $GS_f$, which in this case is 1.  However, it is important to note that the coupling between the randomized execution of $\cA_f$ on $D_1$ and $D_2$ needs to be chosen carefully. For the sake of concreteness, suppose $D_1:= [(\text{Al}, \text{Doctor}),(\text{Ben}, \text{Mechanic}) \allowbreak,(\text{Cal} \allowbreak, \text{Doctor})]$, and $D_2:= [(\text{Ben}, \text{Mechanic}) \allowbreak ,(\text{Cal}, \text{Doctor}), \allowbreak (\text{Dan}, \text{Professor})]$ are neighboring datasets considered in lexicographic order. And let $R = IEI$ be an arbitrary sequence of coin tosses where $I$ means the record was included in the sample and $E$ means the record was excluded from the sample. If we simply choose the identity coupling, then $\vert \cA_f(D_1;R) - \cA_f(D_2;R)\vert = 2$. In general, if $D_1$ alternates between doctors and non-doctors (in lexicographic order) and $D_2$ is obtained by changing the name of the first individual (e.g., Aardvark) so that the individual appears last (e.g., Zuri) then we would have $\vert \cA_f(D_1;IEIE\ldots ) - \cA_f(D_2;IEIE\ldots)\vert = n/2$. Thus, choosing the identity coupling does not give us the tightest upper bound on $CGS_{\cA_f}$ in this case.

To show that $CGS_{\cA_f} \leq 1$ we need to find a coupling ${C \in {\sf Couple}(\cA(D_1),\cA(D_2))}$ such that $(z_1,z_2) \in C$
%permutation $\sigma$ over the set of random coin tosses $r \in \{0,1\}^n$ 
 minimizes the maximum difference of $\vert z_1- z_2\vert$. The key observation here is that $D_1$ and $D_2$ only differ on one entry, so excluding the differing entries in both $D_1,D_2$, we can ``couple'' the random execution for the rest of the entries in $D_1 \setminus \{d^*_1\}$ to match the random execution of the corresponding identical entries in $D_2 \setminus \{d^*_2\}$. In other words, there is some coupling ${C \in {\sf Couple}(\cA(D_1),\cA(D_2))}$ that maintains equivalence (excluding $d^*_1$ and $d^*_2$).

\section{Coupled Global Sensitivity Implies Differential Privacy} \label{sec:cgs}

\begin{proof}[Proof of Theorem~\ref{thm:lap-cgs}]

Let $D_1,D_2 \in \cD$ such that $D_1\sim D_2$ and $\cA: \cD \times \cR \to \bbR^k$. Given $D_1,D_2$, there exists a coupling $C \in {\sf Couple}(\cA(D_1),\cA(D_2))$ such that $\max_{z_1,z_2 \in C} \vert z_1-z_2 \vert \leq CGS_\cA$. Fix an arbitrary point $w \in \bbR^k$, then 
\begin{align*}
&\frac{\Pr[ \cM_L(D_1) =w]}{\Pr[ \cM_L(D_2) =w]}  &\\ 
&= \frac{\Pr_{\{Y_i\}^k_{i=1}}[\cA(D)+(Y_1, \ldots, Y_k)=w]}{\Pr_{\{Y'_i\}^k_{i=1}}[\cA(D')+(Y'_1, \ldots, Y'_k)=w]} &\text{where }Y_i,Y'_i \sim \text{Lap}(CGS_\cA/\epsilon) \\
&=  \frac{\Pr_{(z_1,z_2)\sim C,\{Y_i\}^k_{i=1}}[z_1+(Y_1, \ldots, Y_k)=w]}{\Pr_{(z_1,z_2)\sim C,\{Y'_i\}^k_{i=1}}[z_2+(Y'_1, \ldots, Y'_k)=w]} & \\
&\leq \max_{(z_1,z_2)\sim C}\ \frac{\Pr_{\{Y_i\}^k_{i=1}}[z_1+(Y_1, \ldots, Y_k)=w]}{\Pr_{\{Y'_i\}^k_{i=1}}[z_2+(Y'_1, \ldots, Y'_k)=w]} &\\
    &= \prod^k_{i=1}\left(\frac{\exp \left( - \frac{\epsilon \vert w_i-(z_1)_i \vert }{CGS_\cA}\right) }{\exp \left( - \frac{\epsilon \vert w_i-(z_2)_i \vert }{CGS_\cA}\right)}\right) &\text{applying the def of Laplace distribution}\\
    &= \max_{(z_1,z_2)\sim C}\ \prod^k_{i=1}\exp \left( \frac{\epsilon (\vert w_i- (z_2)_i \vert - \vert w_i-(z_1)_i \vert )}{CGS_\cA}\right) &\\
    &\leq \max_{(z_1,z_2)\sim C}  \prod^k_{i=1} \exp \left( \frac{\epsilon (\vert (z_1)_i - (z_2)_i \vert )}{CGS_\cA}\right) & \text{by triangle inequality}\\ 
    &\leq \max_{(z_1,z_2)\sim C}\ \exp \left( \frac{\epsilon \cdot \| z_1 - z_2 \|_1 }{CGS_\cA}\right) &\\ 
    &\leq \exp(\epsilon)\hspace{0.5cm} & by our assumption
\end{align*}

\end{proof}
In particular, for a randomized algorithm $\cA: \cD\times \cR \to \bbR^k$ the mechanism $\cA(D;R)+(Y_1, \ldots, Y_k)$ is $\epsilon$-differentially private whenever $Y_1,\ldots, Y_k \sim \text{Lap}(CGS_\cA/\epsilon)$ are sampled from the Laplace distribution.   
\end{document}